\documentclass[11pt]{article} 

\usepackage[utf8]{inputenc}
\usepackage[T1]{fontenc}
\usepackage[pdftex,left=1in,top=1in,bottom=1in,right=1in]{geometry}

\usepackage{amsmath, amssymb, dsfont, mathrsfs,amsthm}

\usepackage{graphicx}
\usepackage{caption}
\usepackage{bm}
\usepackage{xcolor}
\usepackage{comment}
\usepackage[colorlinks=true,citecolor=blue,linkcolor=blue]{hyperref}
\hypersetup{
    colorlinks,
    linkcolor={red!50!black},
    citecolor={blue!50!black},
    urlcolor={blue!80!black}
}
\usepackage[nameinlink, noabbrev, capitalize]{cleveref}

\newcommand{\N}{\mathds{N}}
\newcommand{\F}{\mathds{F}}

\newcommand{\R}{\mathbb{R}}
\newcommand{\E}{\mathds{E}}

\newcommand{\cC}{\mathcal{C}}

\newcommand{\set}[1]{\left\{{#1}\right\}}

\newcommand{\wt}{\mathrm{wt}}
\newcommand{\Sp}{\operatorname{Span}}
\newcommand{\Unif}{\mathrm{Unif}}

\newcommand{\supp}{\operatorname{Supp}}

\newcommand{\eps}{\varepsilon}

\newtheorem{theorem}{Theorem}[section]
\newtheorem{proposition}[theorem]{Proposition}

\newtheorem{claim}[theorem]{Claim}
\newtheorem{lemma}[theorem]{Lemma}

\newtheorem{corollary}[theorem]{Corollary}

\newtheorem{definition}[theorem]{Definition}
\newtheorem{remark}[theorem]{Remark}

\newcommand{\poly}{\operatorname{poly}}
\newcommand{\polylog}{\operatorname{polylog}}

\newcommand{\Mod}[1]{\ (\mathrm{mod}\ #1)}

\newcommand{\inparen}[1]{\left(#1\right)}

\let\originalleft\left
\let\originalright\right
\renewcommand{\left}{\mathopen{}\mathclose\bgroup\originalleft}
\renewcommand{\right}{\aftergroup\egroup\originalright}

\title{List-Recovery of Random Linear Codes over Small Fields}

\newif\ifauthors
\authorstrue

\ifauthors{
\author{Dean Doron\thanks{Ben-Gurion University of the Negev. \texttt{deand@bgu.ac.il}.}
\and Jonathan Mosheiff\thanks{Ben-Gurion University of the Negev. \texttt{mosheiff@bgu.ac.il}.}
\and Nicolas Resch\thanks{Informatics Institute, University of Amsterdam. \texttt{n.a.resch@uva.nl}.}
\and João Ribeiro\thanks{Instituto de Telecomunicações and Departamento de Matemática, Instituto Superior Técnico, Universidade de Lisboa. \texttt{jribeiro@tecnico.ulisboa.pt}.}}
}
\else{
}
\fi

\allowdisplaybreaks

\date{}

\begin{document}
	
\maketitle
	
\begin{abstract}

We study list-recoverability of random linear codes over small fields, both from errors and from erasures. 
    We consider codes of rate $\eps$-close to capacity, and aim to bound the dependence of the output list size $L$ on $\eps$, the input list size $\ell$, and the alphabet size $q$. 
    Prior to our work, the best upper bound was $L = q^{O(\ell/\eps)}$ (Zyablov and Pinsker, Prob.\ Per.\ Inf.\ 1981).
    
    Previous work has identified cases in which \emph{linear} codes provably perform worse than non-linear codes with respect to list-recovery. While there exist non-linear codes that achieve $L=O(\ell/\eps)$, we know that $L \ge \ell^{\Omega(1/\eps)}$ is necessary for list recovery from erasures over fields of small characteristic, and for list recovery from errors over large alphabets.

    We show that in other relevant regimes there is no significant price to pay for linearity, in the sense that we get the correct dependence on the gap-to-capacity $\eps$ and go beyond the Zyablov--Pinsker bound for the first time. Specifically, when $q$ is constant and $\eps$ approaches zero,
    \begin{itemize}
        \item For list-recovery from erasures over \emph{prime fields}, we show that $L \leq C_1/\eps$. By prior work, such a result cannot be obtained for low-characteristic fields.

        \item For list-recovery from errors over \emph{arbitrary fields}, we prove that $L \leq C_2/\eps$.
    \end{itemize}
    Above, $C_1$ and $C_2$ depend on the decoding radius, input list size, and field size.
    We provide concrete bounds on the constants above, and the upper bounds on $L$ improve upon the Zyablov--Pinsker bound whenever $q\leq 2^{(1/\eps)^c}$ for some small universal constant $c>0$. 
\end{abstract}

\thispagestyle{empty}

\newpage

\section{Introduction}
Error-correcting codes enable reliable communication over noisy channels by encoding messages $m \in \Sigma^k$ as codewords $c \in \Sigma^n$. A code $\cC \subseteq \Sigma^n$ of rate $R = k/n$ and minimum (relative) Hamming distance $\delta$ allows for reliable communication over an adversarial noisy channel that corrupts up to a $\delta/2$-fraction of codeword symbols.
To tolerate more corruptions, one can relax unique decoding to \emph{list decoding}, where the decoder outputs all codewords within a given Hamming radius. 

This notion is further generalized by \emph{list recovery (from errors)}, which models scenarios where the receiver gets a small list of possible values for each symbol. Formally, a code $\cC$ is said to be \emph{$(\rho, \ell, L)$-list-recoverable} if for every sequence of sets $T_1, \dots, T_n \subseteq \Sigma$ with $|T_i| \leq \ell$, we have
\[
    |\cC \cap B_{\rho}(T_1 \times \cdots \times T_n)| \le L,
\]
where $B_{\rho}(T_1 \times \cdots \times T_n)$ denotes the Hamming ball consisting of all words in $\Sigma^n$ that agree with $T_1, \dots, T_n$ in at least $(1 - \rho)n$ coordinates. When $\ell = 1$, list-recovery from errors reduces to standard list-decoding (from errors).

A related variant is \emph{list recovery from erasures}, where some coordinates are entirely unknown, modeled by setting $T_i = \Sigma$. A code is said to be \emph{$(\alpha, \ell, L)$-list-recoverable from erasures} if
\[
    \left|\cC \cap (T_1\times\dots\times T_n)\right| \leq L
\]
whenever $|T_i| \leq \ell$ for at least $(1 - \alpha)n$ positions. Here too, the case $\ell = 1$ corresponds to list-decoding from erasures.

List recoverable codes are used as a building block for list-decodable and uniquely decodable codes \cite{GI02,GI03,GI04,GI05,KMRSZ17,GKO+18,HRW2020}. They have also gained a significant independent interest, in part due to their applications in pseudorandomness \cite{TZ2004,GUV2009,DMOZ22,KT22}, algorithms (in particular, for heavy hitters, compressed
sensing, and combinatorial group testing \cite{INR10,NPR11,LNN+2016, GLPS17,DW22}), and cryptography \cite{HIOS15,HLR21}.

For both list-recovery from errors and from erasures, there exists a well-defined \emph{capacity threshold} that characterizes the maximal achievable rate for which bounded list-size decoding is possible. Specifically, given parameters $\rho$, $\ell$, and alphabet size $q$, there is a critical rate $R^* = R^*(\rho,\ell,q)$ such that:
\begin{itemize}
    \item For every $\eps > 0$ and any large enough block length, there exist codes of rate $R^* - \eps$ that are $(\rho, \ell, L)$-list-recoverable (from errors or erasures) with
    \begin{equation}\label{eq:ListSize}
        L = O_{\ell,\eps}(1).
    \end{equation}
    \item For every $\eps>0$ and any large enough block length $n$, no code of rate $R^* + \eps$ is $(\rho, \ell, L)$-list-recoverable for $L = q^{o(n)}$.
\end{itemize}
The exact threshold depends on the recovery model:
\begin{itemize}
    \item For $(\rho,\ell,L)$-list-recoverability from \emph{errors}, the threshold rate is
    \[
        R^*_{\mathrm{errors}} = 1 - h_{q,\ell}(\rho),
    \]
    where
    \[
        h_{q,\ell}(\rho) = \rho\,\log_q\!\inparen{\frac{q-\ell}{\rho}} + (1-\rho)\,\log_q\!\inparen{\frac{\ell}{1-\rho}},
    \]
    valid for $0 \le \rho \le 1 - \frac{\ell}{q}$~\cite[Theorem 2.4.12]{Res20}.
    
    \item For $(\alpha,\ell,L)$-list-recoverability from \emph{erasures}, the corresponding threshold is 
    \[
        R^*_{\mathrm{erasures}} = (1 - \alpha)\cdot (1 - \log_q \ell).
    \]
\end{itemize}

The dependence of the list size $L$ on the parameters $\ell$ and $\eps$ (see \cref{eq:ListSize}) is often critical, and has been the focus of extensive research (e.g.,~\cite{RW14,RW18,LP2020,GLS+21,GLMRSW22,GST2022,Tamo2024,CZ2024,LS2025}). 

Using the probabilistic method, it is easy to show that plain random codes achieve $L = O(\ell/\eps)$, and this dependence is often viewed as the optimal benchmark. 
Codes achieving this tradeoff are said to match the \emph{Elias bound for list-recovery}, in reference to the analogous threshold in list-decoding~\cite{Elias1957} (see also~\cite{MRSY24}).

To set expectations for list-recovery of \emph{linear} codes, we recall the classic argument of Zyablov and Pinsker~\cite{ZP81}, adapted to the setting of list-recovery by Guruswami~\cite{Gur03}. Since any set of $L+1$ vectors has a linearly independent subset of size $\log_q (L+1)$, and since the events that linearly independent vectors lie in a linear code are stochastically independent, the argument for plain random codes gives $L = q^{O(\ell/\eps)}$. 
This naturally raises the question: what is the actual \emph{price of linearity} in list-recovery? While various forms of degradation are possible, we seek to understand how the requirement of linearity affects the list-size achievable near capacity.

\paragraph{Prior work}
Most previous results concerning the output list size have focused on the \emph{large alphabet regime}, where $q$ is at least exponential in $1/\eps$. In this setting, Li and Shagrithaya \cite{LS2025} recently proved that random linear codes are almost surely list-recoverable with list size $L \le \inparen{\ell/\eps}^{O\inparen{\ell/\eps}}$. By a reduction between code ensembles~\cite{LMS24}, the same upper bound also holds with high probability for Reed--Solomon codes over random evaluation sets. On the other hand, \cite{LS2025} proved that all $(\rho,\ell,L)$-list-recoverable (from errors) linear codes over a large alphabet must satisfy $L \ge \ell^{\Omega(R/\eps)}$, implying an exponential gap in the list-size between linear codes and plain random codes. A similar negative result was previously proven by Chen and Zhang \cite{CZ2024} for Reed-Solomon codes and folded Reed--Solomon codes. Recently, Komech and Mosheiff \cite{KM2025} constructed a new ensemble of non-linear codes that achieve $L\le O(\ell/\eps)$ in list-recovery from errors. These are the only codes other than plain random codes known to achieve the list-recovery Elias bound.

Older works of Rudra and Wootters~\cite{RW14,RW18}, incomparable to \cite{LS2025}, show (in our terms) that random linear codes achieve  list-size $L \le \ell^{\frac 1\eps\cdot \log^2\inparen{\ell/\eps}}$ for list-recovery (from errors) in the large alphabet regime, but only under the guarantee that $\rho = 1- \Omega(\eps)$.

Other works concern the list-recoverability of folded Reed--Solomon codes, multiplicity codes, tensor codes, and variants of them (e.g., \cite{KRS+2018,HRW2020,Tamo2024}). In particular, Tamo~\cite{Tamo2024} shows that folded Reed--Solomon codes and multiplicity codes achieve list size $L\le \inparen{\ell/\eps}^{O\inparen{\frac {1+\log \ell}\eps}}$ (in the errors case) in the large alphabet regime.

In contrast, very little is known about list-recovery in the \emph{small alphabet regime}, where $q$ is sub-exponential in $1/\eps$. To the best of our knowledge, the only positive result in this setting is the aforementioned $L=q^{O(\ell/\eps)}$ for random linear codes \cite{ZP81}. On the other hand, we know that when $q$ is a power of a small prime, random linear codes in this regime are very unlikely to be $(\alpha,\ell,L)$-list-recoverable from erasures with $\ell^{o(1/\eps)}$~\cite{GLMRSW22} (we state this as \Cref{thm:list-rec-erasures-lb} below). 

Provable separation between linear and nonlinear codes were also established in the setting of list decoding from erasures. 
For the special case of $\F_2$, and aiming for an erasure decoding radius of $\alpha = 1-\eps$, Guruswami \cite{Gur03} showed that the output list size must satisfy $L=\Omega(1/\eps)$, as long as the rate is sufficiently non-vanishing. In contrast, plain random codes achieve $L=O(\log(1/\eps))$, and we even have explicit constructions that nearly match plain random codes in some parameter regimes \cite{BDT20}.

\subsection{Our Contribution}

We establish new upper bounds on the output list size for list-recovery of random linear codes near capacity in the small alphabet regime. To the best of our knowledge, these are the only known bounds for general list-recovery in this setting beyond the classical Zyablov--Pinsker argument. Notably, our results achieve list sizes with only linear dependence on \(1/\eps\), in contrast to the exponential dependence in previous bounds.

\paragraph{List-recovery from erasures over prime fields.}
For our first result, we consider $(\alpha,\ell,L)$-list-recovery from erasures.
Recall that the capacity in this case is 
\[
    R^*_{\mathrm{erasures}}(\alpha,\ell,q) = (1-\alpha)(1-\log_q\ell).
\]
As mentioned before,~\cite{GLMRSW22} showed that there is a price to pay for linearity over fields of small characteristic.
More precisely, they proved the following.
\begin{theorem}[\protect{informal; see~\cite[Theorem III.1]{GLMRSW22}}] \label{thm:list-rec-erasures-lb}
    If $\ell$ divides $\mathrm{char}(\F_q)$, then with high probability over the choice of the linear code, the output list size $L$ cannot be taken smaller than $\ell^{\Omega(1/\eps)}$.
\end{theorem}

We show that this limitation disappears when considering a \emph{prime} field size $q$.
In this case, we can make the output list size $O_{\alpha,\ell,q}(1/\eps)$.
\begin{theorem}[list recovery from erasures over prime fields; see \cref{thm:erasures}]\label{thm:erasures-intro}
    Given $1 \leq \ell \leq q$ with $q$ prime and $\alpha \in [0,1)$, there exists $C(\alpha,\ell,q)>0$ such that the following holds. Let $\cC \leq \F_q^n$ be a random linear code of rate $R^*_{\mathrm{erasures}}(\alpha,\ell,q)-\eps$ for some $\eps > 0$. Then, $\cC$ is with high probability $\left(\alpha,\ell,\frac{C(\alpha,\ell,q)}{\eps}\right)$-list-recoverable from erasures. 
\end{theorem}

While our focus is on the setting where $\alpha$ and $q$ are constants, we determine effective bounds on $C(\alpha,\ell,q)$ even when $\alpha$ is a function of $n$, and $q$ is slightly super-constant.\footnote{For a concrete example, when $\ell \le (1-\gamma)q$ for some constant $\gamma$, and $\alpha$ is bounded away from $1$, we can take $C(\alpha,\ell,q) \le q^{O(\ell \cdot \log q)}$ as long as $n \ge C(\alpha,\ell,q)$. We refer the reader to \Cref{thm:erasures} for the precise bound.}

\paragraph{List-recovery from errors.}

Next, we consider the case of list-recovery from errors, where we recall the capacity is 
\[
    R^*_{\mathrm{errors}}(\rho,\ell,q) = 1 - h_{q,\ell}(\rho).
\]

Here, contrary to what happens in the regime of large $q$-s, we do not observe any price to pay for linearity, at least in terms of the dependence on the gap-to-capacity $\eps$. 

\begin{theorem}[list recovery from errors; see \cref{thm:errors}]\label{thm:errors-intro}
    Given $1 \leq \ell \leq q$ with $q$ a prime power and $\rho \in (0,1-\ell/q)$, there exists $C(\rho,\ell,q)>0$ such that the following holds. Let $\cC \leq \F_q^n$ be a random linear code of rate $R^*_{\mathrm{errors}}(\rho,\ell,q)-\eps$ for some $\eps>0$. Then, $\cC$ is with high probability $\left(\rho,\ell,\frac{C(\rho,\ell,q)}{\eps}\right)$-list-recoverable from errors. 
\end{theorem}

Again, we determine effective bounds on the numerator $C(\rho,\ell,q)$, which can be found in \Cref{thm:errors}.\footnote{Specifically, when $\rho$ is bounded away from both $0$ and $1-\ell/q$, we can take $C(\rho,\ell,q) \le q^{\ell \cdot \polylog(q)}$. See \cref{thm:errors} for the precise bound.} For context, we recall that in the ``large $q$ regime'' (say, $q \gg \ell^{1/\eps}$), such a result is provably impossible. Specifically, when $q$ is large, the list-recovery capacity is (essentially) the Singleton bound, namely, $R^*_{\mathrm{errors}} \approx 1-\rho$. It is known that at rate $1-\rho-\eps$, \emph{any} linear code list-recoverable from errors must have $L \geq \ell^{\Omega(R/\eps)}$ \cite{LS2025}. That is, the dependence of $L$ on the gap-to-capacity \emph{must be exponential}. But in our case, at least if $\rho,\ell$ and $q$ are held constant, the dependence of $L$ on $\eps$ is just $O(1/\eps)$.

We conclude this section with some remarks. 

\begin{remark}[on the field insensitivity]
    \em
    Both of our results show that there is no significant price to pay for linearity in list-recovery for certain parameter regimes.
    Our result in \cref{thm:errors-intro} on list-recovery from errors is insensitive to the field size.
    On the other hand, as discussed above, our result in \cref{thm:erasures-intro} on list-recovery from erasures is (necessarily!) field sensitive.
    We provide some informal discussion on why this happens.
    First, note that we work with rates close to capacity, and the capacity in the erasures setting is larger for comparable $\alpha$ and $\rho$.
    Second, looking ahead (see \cref{sec:tech-overview} for more details), list-recovery from erasures depends on the ``additive structure'' of $T_1\times \cdots\times T_n$, for arbitrary size-$\ell$ subsets $T_1,\dots,T_n\subseteq \F_q$.
    If the $T_i$'s are \emph{subspaces} of $\F_q$, then it is quite likely these form a bad configuration for list-recovery from erasures.
    In contrast, list-recovery from errors depends on the additive structure of the ``puffed-up'' combinatorial rectangles
    \begin{equation*}
        B_\rho(T_1 \times T_2 \times \cdots \times T_n) = \{x \in \F_q^n:x_i \notin T_i \text{ for at most }\rho n \text{ choices of } i \in [n]\}.
    \end{equation*}
    Even if the $T_i$'s are subspaces of $\F_q$, the ``puffing up'' operation kills any additive structure that could lead to a bad list-recovery configuration.
\end{remark}

\begin{remark}[on the dependence of $L$ on the various parameters]
    \em
    Recall that a code which is $\eps$-close to capacity is said to achieve the Elias bound if $L = O(\ell/\eps)$. Note that in both \Cref{thm:erasures-intro} and \Cref{thm:errors-intro}, the dependence of $L$ on $\eps$ is as we would hope. However, the dependence on the other parameters (particularly $\ell$ and $q$) is exponentially worse than the $O(\ell)$ dependence. We leave it as a natural open problem to improve on this dependency.
\end{remark}

\begin{remark} [comparison to~\cite{ZP81}]
    \em 
    As mentioned earlier, we are not aware of any prior arguments establishing non-trivial bounds on the list-size $L$ for codes $\eps$-close to capacity which do not require $q$ to be large, other than what follows from the Zyablov--Pinsker argument~\cite{ZP81} (given formally in~\cite{Gur03}). Recall that this method guarantees list size $L=q^{O(\ell/\eps)}$. In comparison, we obtain (roughly) $L = \frac{1}{\eps} \cdot q^{O(\ell \log q)}$ in the case of erasures (over prime fields), and $L = \frac{1}{\eps} \cdot q^{\ell \cdot \polylog(q)}$ in the case of errors (over all fields, assuming $\rho$ is not too close to $0$ or $1-\ell/q$). 
    Thus, compared to \cite{ZP81}, we obtain a smaller bound on $L$ once, roughly, $\eps < \frac{1}{\polylog(q)}$.
    In particular, when $\ell$ and $q$ are constants we achieve asymptotic improvement in $L$. 
\end{remark}

\subsection{Technical Overview}\label{sec:tech-overview}

At a conceptual level, our work reconsiders the approach of Guruswami, H{\aa}stad, and Kopparty \cite{GHK11}, which allowed for an understanding of the list-decodability from errors of random linear codes over constant-sized alphabets, and adapts it to the case of list-recovery (either from erasures or errors). Recall that list-decoding from errors is the special case of list-recovery from errors with $\ell=1$. 

Using terminology that we define in this work, the first step in the argument of \cite{GHK11} is to argue that Hamming balls are nontrivially \emph{mixing}. Specifically, fix any center $z \in \F_q^n$ and consider sampling twice uniformly and independently from the Hamming ball of radius $\rho$ centered at $z$; denote the two samples by $X$ and $X'$. The authors argue that for some $\delta>0$ and any $\alpha,\beta \in \F_q\setminus\{0\}$, it follows that for any other center $y \in \F_q^n$, $\Pr[\alpha X+\beta X' \in B_\rho(y)] \leq q^{-\delta n}$ for some $\delta = \delta(\rho,q)>0$.\footnote{In fact for \cite{GHK11} it sufficed to only consider the $z=0$ case, in which case one can always assume $\alpha=\beta=1$. But their argument naturally generalizes to this case, and this is the notion of mixing that we require for our list-recovery results.} Alternatively, we could say that for any $y \in \F_q^n$, we have $\Pr[\alpha X+\beta X' \in y + B_\rho(z)] \leq q^{-\delta n}$. From here, using some additional tools like \emph{2-increasing chains} -- whose existence they establish via a Ramsey-theoretic argument -- they are able to show that random linear codes with $\eps$ gap-to-capacity are with high probability $(\rho,\frac{C(\rho,q)}{\eps})$-list-decodable from errors. In particular, their techniques promise the correct dependence on $\eps$ (although the dependence on the parameters $q$ and $\rho$ is quite poor).

We recommend viewing this ``$\delta$-mixing'' property in the following light. Any argument establishing that random linear codes have good list-decodability must somehow argue that random subspaces and Hamming balls don't tend to ``correlate'' too much. In particular, it should not be the case that Hamming balls have noticeable ``linear structure'', and in particular they should be ``far'' from being closed under addition.

\medskip

In our work, we consider whether or not sets that are relevant for list-recovery, i.e., the sorts of sets that list-recoverable codes cannot intersect with too much, also have nontrivial mixing. 
Firstly, we crystallize in a definition what it means for any arbitrary set $T \subseteq \F_q^n$ to be \emph{$\delta$-mixing}, where $\delta>0$: for any $\alpha,\beta \in \F_q\setminus\{0\}$ and $y \in \F_q^n$, if $X,X' \sim T$ -- which denotes that $X$ and $X'$ are sampled independently and uniformly from $T$ -- then $\Pr[\alpha X+\beta X' \in y+T] \leq q^{-\delta n}$. 

Now, for $(\alpha,\ell,L)$-list-recovery from erasures the relevant sets are combinatorial rectangles $T_1 \times T_2 \times \cdots \times T_n$ where for at least $(1-\alpha)n$ values of $i \in [n]$ we have $|T_i| \leq \ell$. For $(\rho,\ell,L)$-list-recovery from errors the sets are ``puffed-up'' combinatorial rectangles. Namely, for $T_1,T_2,\dots,T_n \subseteq \F_q$ each with $|T_i| \leq \ell$, we consider \emph{list-recovery balls}
\[
	B_\rho(T_1 \times T_2 \times \cdots \times T_n) = \{x \in \F_q^n:x_i \notin T_i \text{ for at most }\rho n \text{ choices of } i \in [n]\}.
\]
Following the argument of \cite{GHK11}, once we establish that these sets are nontrivially mixing, we can obtain bounds on the list-size with the correct dependence on $\eps$. Our task then boils down to understanding the mixingness of the sets relevant for list-recovery. We consider first the erasures case, and subsequently discuss the errors case. 

\paragraph{List-recovery from erasures} Firstly, observe that for a combinatorial rectangle $T_1 \times \cdots \times T_n$, if each of the sets $T_i$ is nontrivially mixing as a subset of $\F_q$ (take the $n=1$ case of the above definition), then $T_1 \times \cdots \times T_n$ is also nontrivially mixing (as a subset of $\F_q^n$). Hence, it suffices for us to consider whether or not subsets of $\F_q$ mix. It is here that the dependence on the field size shows up.

Recall from the earlier discussion that \cite{GLMRSW22} established that random linear codes over $\F_{\ell^t}$ (where $\ell$ is a prime power) are with high probability \emph{not} $(\alpha,\ell,L)$-list-recoverable from erasures unless $L \geq \ell^{\Omega(1/\eps)}$. Indeed, it is easy to find a subset $T \subseteq \F_{\ell^t}$ which is \emph{not} $\delta$-mixing for any $\delta>0$: take $T= \F_\ell$ (or, more generally, any multiplicative coset $\gamma \cdot \F_\ell$ for $\gamma \in \F_{\ell^t}\setminus \{0\}$). Since $\F_\ell$ is closed under addition, in particular we have that if $X$ and $X'$ are sampled independently and uniformly from $T$ then $\Pr[X+X' \in T] = 1$, so $T$ is \emph{not} $\delta$-mixing for any $\delta>0$. 

What went wrong in this example? The fact that $\F_{\ell^t}$ contains $\F_\ell$ as a subfield means that $\F_{\ell^t}$ contains non-trivial $\F_\ell$-linear subspaces. Such subspaces naturally create ``bad'' input lists, and the argument of \cite{GLMRSW22} establishes that indeed a random linear code is likely to contain many vectors from a combinatorial rectangle $T_1 \times \cdots \times T_n$ where at least a $(1-\alpha)$ fraction of the $T_i$'s are $1$-dimensional $\F_\ell$-subspaces of $\F_{\ell^t}$. 

If we insist that $q$ be a \emph{prime}, then $\F_q$ \emph{does not} have any non-trivial subspaces. 
However, in this case $\F_q$ still contains some subsets with ``additive structure;'' for example, taking $\ell$ to be odd here for simplicity, centered intervals\footnote{Or, more generally, arithmetic progressions.} like $I = \{-\frac{\ell-1}{2},-\frac{\ell-1}{2}+1\dots,\frac{\ell-1}{2}-1,\frac{\ell-1}{2}\} \subseteq \F_q$
have the property that if one samples $X,X' \sim I$ independently, then $\Pr[X+X' \in I] \approx \frac{3}{4}+\frac{1}{4\ell^2}$ assuming $\ell \leq 2q/3$.
But note that this is still non-trivially bounded away from $1$! Remarkably, an argument of Lev~\cite{Lev01} shows that this is the worst case over prime fields.
More precisely, over all sets $T$ of size $\ell$, to maximize $\Pr[\alpha X + \beta X' \in \gamma+T]$ where $\alpha,\beta,\gamma \in \F_q$ with $\alpha,\beta \neq 0$, one should choose $T=I$ (the centered interval of length $\ell$), $\alpha=\beta=1$ and $\gamma=0$. 
In our technical section we give an effective bound on $\Pr[X+X' \in I]$ for all $\ell<q$, which allows us to argue that combinatorial rectangles $T_1 \times \cdots \times T_n$ in which at least a $1-\alpha$ fraction of the $T_i$'s have size at most $\ell$ are nontrivially mixing, as desired.

\paragraph{List-recovery from errors} We now wish to establish that list-recovery balls are nontrivially mixing. Notably, unlike in the case of erasures, the argument here is \emph{insensitive} to the base field.
Let $T = T_1 \times \cdots \times T_n$, where each $|T_i| \leq \ell$. Let $X,X' \sim B_\rho(T)$; in this overview, we will sketch how one bounds $\Pr[X+X' \in B_\rho(T)]$ (the argument easily generalizes to allow for multipliers $\alpha,\beta \in \F_q\setminus\{0\}$ and a shift $y \in \F_q^n$). 

Unlike in the case of combinatorial rectangles, it is \emph{not} the case that $X=(X_1,\dots,X_n)$ and $X' = (X_1',\dots,X_n')$ have independent coordinates. For example, conditioned on $X_1$ lying in $T_1$, then $X_2$ is less likely to lie in $T_2$. However, these correlations are relatively minor, and we can essentially ``pretend'' that both $X$ and $X'$ are sampled as follows: for each $i \in [n]$, with probability $1-\rho$ set the $i$-th coordinate to a uniformly random element of $T_i$, and otherwise set it to a uniformly random element of $\F_q\setminus T_i$, and these choices are made independently for each $i \in [n]$.\footnote{In fact for technical reasons we have to consider $X_i$ lying in $T_i$ with probability $\omega$ for $\omega \leq \rho$, and similarly $X_i'$ lies in $T_i$ with probability $\omega' \leq \rho$. But by a concentration argument one can easily establish that $\omega$ and $\omega'$ are with high probability very close to $\rho$.} We remark that a similar trick is implicit in~\cite{GHK11}, and made explicit in the context of rank-metric codes by Guruswami and Resch~\cite{GR18}. 

This new distribution is much more amenable to analysis. In particular, letting $E_i$ be the indicator for the event $X_i+X_i' \in T_i$, then $X+X' \in B_\rho(T)$ iff $\sum_i E_i \geq (1-\rho)n$. Thus, if we can argue $\Pr[E_i=1] < 1-\rho$ then a classic Chernoff-Hoeffding bound establishes $\Pr\left[\sum_i E_i \geq (1-\rho)n\right]$ is exponentially small, implying the desired $\delta$-mixing. 

Bounding $\Pr[E_i=1]$ is the most novel part of the analysis, and is done in \Cref{lem:mixing-errors-helper}. Recall that $X_i,X_i' \sim T_i$, and let also $Y_i,Y_i' \sim \F_q\setminus T_i$. We have 
\begin{align}
	\Pr&[E_i=1] = (1-\rho)^2\Pr[X_i+X_i' \in T_i] + 2\rho(1-\rho)\Pr[X_i+Y_i' \in T_i] + \rho^2\Pr[Y_i+Y_i' \in T_i] \nonumber \\
	&= \sum_{z_i \in T_i}(1-\rho)^2\Pr[X_i+X_i'=z_i] + 2\rho(1-\rho)\Pr[X_i+Y_i'=z_i] + \rho^2\Pr[Y_i+Y_i'=z_i] \ . \label{eq:intro-interesting-term}
\end{align}
As is standard, $\Pr[X_i+X_i'=z_i]$ is proportional to $1_{T_i}*1_{T_i}(z_i)$, the convolution of the indicator functions for $T_i$, and similarly $\Pr[X_i+Y_i'=z_i]$ and $\Pr[Y_i+Y_i'=z_i]$ are proportional to $1_{T_i} * 1_{\F_q\setminus T_i}(z)$ and $1_{\F_q\setminus T_i}*1_{\F_q\setminus T_i}(z)$, respectively. Using the simple identity $1_{\F_q\setminus T_i} = 1-1_{T_i}$, we can rewrite \cref{eq:intro-interesting-term} as 
\[
	c \cdot 1_{T_i}*1_{T_i}(z) + \text{other terms} \ ,
\]
where $c$ is a constant which we show to be positive assuming $\rho \in (0,1-\ell/q)$ (and indeed, if $\rho > 1-\ell/q$ then it could be negative). Upon giving a trivial upper bound $1_{T_i}*1_{T_i}(z_i)$ (which corresponds to the case that $T_i$ does not mix at all, as we must do since we are making no assumptions on the field) and simplifying, we obtain the bound
\[
	\Pr[E_i=1] \leq (1-\rho)^2 + \rho^2 \cdot \frac{\ell}{q-\ell}.
\]
To our satisfaction, we have that $(1-\rho)^2 + \rho^2 \cdot \frac{\ell}{q-\ell} < 1-\rho$ iff $\rho \in (0,1-\ell/q)$, which is precisely the range of decoding radius at which we can hope for positive rate list-recovery from errors in the first place! Thus, we have that $\Pr\left[\sum_i E_i \geq (1-\rho)n\right]$ is exponentially small, establishing that list-recovery balls nontrivially mix. 

\subsection{Open Problems}

Lastly, we leave here some directions for future research:
\begin{itemize}
    \item In our results the dependency on $\eps$ is correct, but the dependency on $\ell$ and $q$ is rather poor. Can we improve this dependency? Or, can we perhaps prove new lower bounds on $L$ in terms of $\ell$ and $q$ that apply when these parameters are not too big?

    \item \cite{GLMRSW22} showed that over $\F_q$ with $q=\ell^t$ (and hence small characteristic), a random linear code is with high probability \emph{not} $(\alpha,\ell,L)$-list-recoverable from erasures unless $L\geq \ell^{\Omega(1/\eps)}$.
    Can we show that this lower bound actually applies to \emph{every} $q$-ary linear code?

    \item Many arguments for codes being list-recoverable from errors in fact establish the stronger property of \emph{average-radius}-list-recovery, where now one instead shows that for any input lists $T_1,\dots,T_n \subseteq \F_q$ of size $\ell$, given $L+1$ codewords $c^{(1)},\dots,c^{(L+1)}$ one has
    \[
        \frac{1}{L+1}\sum_{j=1}^{L+1} 1_{\{c^{(j)}_i \notin T_i\}}  > \rho n.
    \]
    This in particular implies that there cannot be $L+1$ codewords lying in a list-recovery ball of radius $\rho$. We believe our method should be able to establish this (slightly) stronger guarantee for random linear codes.
\end{itemize}

\section{Preliminaries}

\subsection{Notation}

We will often denote random variables and sets by uppercase Roman letters.
The distinction will be clear from context.
We write $[n]=\{1,\dots,n\}$ for any positive integer $n$.
For a vector $v\in\F_q^n$ with $\F_q$ the finite field of order $q$, we write $\supp(v)=\{i:v_i\neq 0\}$. We define the \emph{weight} of $v$ to be $\wt(v)=|\supp(v)|$, and the \emph{Hamming distance} between $v$ and $u$ is $d(u,v)=|\{i \in [n]:u_i \neq v_i\}| = |\supp(u-v)|$. 
For a collection of vectors $v_1,\dots,v_d\in\F_q$, we denote by $\Sp(v_1,\dots,v_d)$ the subspace of $\F_q^d$ spanned by $v_1,\dots,v_d$. 

We denote the \emph{binary entropy function} by $h_{2}(\cdot)$, and recall that
$h_{2}(x) = x\log_{2}\frac{1}{x} + (1-x)\log_{2}\frac{1}{1-x}$. For a positive integer $q$ we shorthand $\exp_q(x) = q^x$. Additionally, by default $\log$ is the base-2 logarithm.
We write $1_{A}$ for the indicator function of a set $A$, and
write $1_{\{E\}}$ for the indicator random variable that equals $1$ if and only if the event $E$ holds. 

\subsection{The Random Code Model}

For an alphabet size $q$, a \emph{plain random code} $\cC$ of block length $n$ and rate $R \in [0,1]$ is obtained by including each $x \in [q]^n$ in $\cC$ with probability $q^{-(1-R)n}$, and these choices are made independently for each $x$. (Note then that such a code has size $q^{Rn}$ in expectation, and by a Chernoff bound it follows that it has rate $R\pm o(1)$ with high probability.)

When $q$ is a prime power, a \emph{random linear code} $\cC$ of block length $n$ and rate $R \in [0,1]$ is obtained by sampling a uniformly random matrix $H \in \F_q^{(1-R)n \times n}$, and defining
\[
    \cC = \{x \in \F_q^n : Hx = 0\}.
\]
Note that $H$ is full-rank with probability $1-O(q^{-Rn})$, and therefore $\cC$ has rate $R$ with high probability. 

Given any subset $S \subseteq [q]^n$, if $\cC$ is a plain random code of rate $R$ then $\Pr[S \subseteq \cC] = q^{-(1-R)n|S|}$. When dealing with random linear codes, the probability that a set appears in the code is determined by the span of the set.

\begin{proposition} \label{prop:prob-contained-in-rlc}
    Let $\cC \leq \F_q^n$ be a random linear code of block length $n$ and rate $R \in [0,1]$, and let $x_1,\dots,x_b \in \F_q^n$. Then,
    \[
        \Pr[\forall i \in [b], ~x_i \in \cC] = q^{-(1-R)n\dim(\Sp(v_1,\dots,v_b))}.
    \]
\end{proposition}

\subsection{List-Recovery Notions}

This section collects the basic notions of list-recovery we study. 

\paragraph{List-recovery from erasures.} We begin with the relevant definition. 

\begin{definition}[list recovery from erasures]
Let $\mathcal{C} \subseteq [q]^n$ be a $q$-ary code of block-length $n$. For an erasure radius $\alpha \in [0,1)$ and input list size $1 \le \ell \le q$, we say that $\mathcal{C}$ is 
$(\alpha,\ell,L)$\emph{-list-recoverable from erasures} if for every $T_1,\ldots,T_n \subseteq [q]$ such that $|T_i| \le \ell$ for at least $(1-\alpha)n$ of the $i$-s and $T_i = [q]$ for the remaining, it holds that
    \[
        |\cC \cap (T_1 \times \cdots \times T_n)| \leq L.
    \]
    That is, in any combinatorial rectangle of which at least $(1-\alpha)n$ of its side-lengths are at most $\ell$ (and the remainder can be as large as $q$), there are at most $L$ codewords. 
\end{definition}

We will also consider list-recovery from errors. The concept of a list-recovery ball -- which generalizes that of a Hamming ball -- will be useful.

\begin{definition}[list-recovery ball] \label{def:list-rec-ball}
    Let $q \in \mathbb N$ and let $T_1,\dots,T_n \subseteq [q]$. The \emph{list-recovery ball of radius $\rho$ centered at $T_1 \times \cdots \times T_n$} is
    \[
        B_\rho(T_1\times \cdots \times T_n) = \{x \in [q]^n:d(x,T_1 \times \cdots \times T_n) \leq \rho n\}.
    \]
    Above, we have extended the Hamming metric by setting 
    \[
        d(x,T_1 \times \cdots \times T_n)=|\{i\in[n]:x_i \notin T_i\}|.
    \]
\end{definition}

We now state the relevant capacity theorem for list-recovery from erasures. The proofs of the two implications are standard: the \emph{possibility} result follows from analyzing the performance of a plain random code, while the \emph{impossibility} result follows from a counting argument. 

\begin{theorem} [list-recovery from erasures capacity]\label{thm:ErasuresCap}
    Let $1 \leq \ell \leq q$ and let $\alpha \in [0,1)$. Fix $\eps>0$. For $n \in \mathbb N$ large enough, the following hold:
    \begin{itemize}
        \item There exists a code $\cC \subseteq [q]^n$ of rate $1-(1-\alpha)\log_q\ell-\alpha-\eps$ which is $(\alpha,\ell,\lceil\ell/\eps\rceil)$-list-recoverable from erasures.
        
        \item For any code $\cC \subseteq [q]^n$ of rate $1-(1-\alpha)\log_q\ell -\alpha+\eps$, there exist $T_1,\dots,T_n \subseteq [q]$ with $|T_i| \leq \ell$ for at least $(1-\alpha)n$ values of $i \in [n]$ such that $|\cC \cap (T_1 \times \cdots \times T_n)| \geq q^{\eps n - o(n)}$. 
    \end{itemize}
\end{theorem}

We therefore say that the \emph{capacity} for $(\alpha,\ell,L)$-list-recovery from erasures is $1-\alpha-(1-\alpha)\log_q\ell$. We will study what happens for codes of rate $1-(1-\alpha)\log_q\ell-\alpha-\eps$ for some $\eps>0$, and determine the value of their output list-size $L$. We will be focused on the case where $q$ and is held constant, and the gap-to-capacity $\eps$ tends to $0$. 

\paragraph{List-recovery from errors.} We now define what it means for a code to be list-recoverable from errors. 

\begin{definition} [list recovery from errors]
    Let $\mathcal{C} \subseteq [q]^n$ be a $q$-ary code of block-length $n$. For 
    a decoding radius $\rho \in (0,1-\ell/q)$ and
    input list size $1 \leq \ell \leq q$, we say that $\mathcal{C}$ is  $(\rho,\ell,L)$\emph{-list-recoverable from errors} if for every $T_1,\ldots,T_n \subseteq [q]$ such that $|T_i| \le \ell$ for all $i\in[n]$, it holds that
    \[
        |\cC \cap B_\rho(T_1\times\cdots\times T_n)| \leq L.
    \]
    That is, every list-recovery ball of radius $\rho$ with side-lengths $\leq \ell$ contains at most $L$ codewords. 
\end{definition}

We will need an estimate on the size of list-recovery balls. It makes use of the \emph{$(q,\ell)$-entropy function}, defined as follows:
\begin{align}
    h_{q,\ell}(x) =x\log_q \frac{q-\ell}{x} + (1-x)\log_q\frac{q-\ell}{1-x}. \label{eq:ql-entropy-def}
\end{align}
An operational interpretation of this quantity is as the base-$q$ entropy of a random variable which, with probability $1-x$, samples a uniformly random element from a set of size $\ell$, and with probability $x$ samples a uniformly random element from the complement. Additionally, so long as $0<x < 1-\ell/q$ it holds that $0 < h_{q,\ell}(x) < 1$. Note that if $\ell=1$ one recovers the $q$-entropy function, which we denote as $h_q$ (i.e., if $\ell$ is omitted from the subscript, then it is by default $1$). 

We now state the relevant estimate. 

\begin{proposition}[{\cite[Proposition~2.4.11]{Res20}}] \label{prop:estimate}
    Let $1 \leq \ell \leq q$ be integers and $\rho \in (0,1-\ell/q)$. Let $T_1,\dots,T_n \subseteq [q]$ with $|T_i| = \ell$ for all $i \in [n]$. Then, 
    \[
        \frac{q^{n h_{q,\ell}(\rho)}}{\sqrt{2n}} \leq |B_\rho(T_1 \times \cdots \times T_n)| \leq q^{nh_{q,\ell}(\rho)}.
    \]
\end{proposition}

\begin{remark}
    \em
    In \cite{Res20} the lower bound is just given as $q^{nh_{q,\ell}(\rho)-o(n)}$. Using an estimate of Wozencraft and Reiffen~\cite{WR61} (which we use below in the proof of \Cref{prop:prob-bin-hits-expectation}), it follows that one can take the $q^{-o(n)}$ term to be $\frac{1}{\sqrt{2n}}$. 
\end{remark}

This estimate drives the following capacity theorem.

\begin{theorem} [list-recovery from errors capacity]
    Let $1 \leq \ell \leq q$ and let $\rho \in (0,1-\ell/q)$. Fix $\eps>0$. For $n \in \mathbb N$ large enough, the following hold:
    \begin{itemize}
        \item There exists a code $\cC \subseteq [q]^n$ of rate $1-h_{q,\ell}(\rho)-\eps$ which is $(\rho,\ell,\lceil\ell/\eps\rceil)$-list-recoverable from errors.
        \item For any code $\cC \subseteq [q]^n$ of rate $1-h_{q,\ell}(\rho)+\eps$, there exists $T_1,\dots,T_n \subseteq [q]$ with $|T_i| \leq \ell$ for all $i \in [n]$ such that $|\cC \cap B_\rho(T_1 \times \cdots \times T_n)| \geq q^{\eps n - o(n)}$. 
    \end{itemize}
\end{theorem}

Thus, we will concern ourselves with codes of rate $1-h_{q,\ell}(\rho)-\eps$, and determine the output list size $L$ for $(\rho,\ell,L)$-list-recovery from errors. And, as in the list-recovery from erasures case, we will hold $\rho,\ell,q$ constants and consider the asymptotic behaviour of $L$ as the gap-to-capacity $\eps \to 0$.

We will also need the following lower bound on the difference $h_{q,\ell}(\rho)-h_{q,\ell}(\rho-\eta)$.
\begin{claim}\label{lem:bound-diff-entropy}
    For any integers $q>0$ and $1\leq \ell\leq q$, any $\rho\in(0,1-\ell/q]$, and any $\eta\in[0,\rho]$, we have
    \begin{equation*}
        h_{q,\ell}(\rho)-h_{q,\ell}(\rho-\eta) \geq \eta \log_q\left(\frac{(q-\ell)(1-\rho)}{\ell \cdot \rho}\right)\geq 0.
    \end{equation*}
\end{claim}
\begin{proof}
    We can rewrite
    \begin{equation*}
        h_{q,\ell}(\rho)-h_{q,\ell}(\rho-\eta) = \frac{h_2(\rho)-h_2(\rho-\eta)}{\log q} + \eta \log_q\left(\frac{q-\ell}{\ell}\right).
    \end{equation*}
    Now, the desired lower bound follows if we show that
    \begin{equation*}
        h_2(\rho)-h_2(\rho-\eta) \geq \eta \log\left(\frac{1-\rho}{\rho}\right)
    \end{equation*}
    for all $\rho\in[0,1]$ and $\eta\in[0,\rho]$.
    To see this, first note that $h_2'(\rho)= \log\left(\frac{1-\rho}{\rho}\right)$.
    Since $h_2$ is concave in $[0,1]$, we have $h_2(\rho)-h_2(\rho-\eta) \geq \eta\cdot  h'_2(\rho)$, as desired.
    
    Finally, the rightmost lower bound ``$\geq 0$'' holds under the lemma's hypotheses since $\frac{(q-\ell)(1-\rho)}{\ell \cdot \rho}\geq 1$ if and only if $\rho\leq 1-\ell/q$.
\end{proof}

\subsection{Some Distributions and Concentration Bounds}

For a given finite set $S$ in some universe $U$, we write $X \sim S$ to denote that $X$ is distributed uniformly over the set $S$. 
For $p \in [0,1]$, we denote by $\mathrm{Ber}(p)$ the distribution of a Bernoulli random variable, which equals $1$ with probability $p$ and equals $0$ with probability $1-p$. By $\mathrm{Bin}(n,p)$, we denote a binomial random variable which is the sum of $n$ i.i.d.\ $\mathrm{Ber}(p)$ random variables. 

We now provide a tail bound on $\mathrm{Bin}(n,p)$ random variables. It makes use of the concept of a Kullback--Leibler divergence between two Bernoulli random variables. Specifically, for $x, y \in [0,1]$, define
\[
    D(x\|y) =x\ln\frac xy + (1-x)\ln \frac{1-x}{1-y}.
\]
\begin{proposition}\label{prop:2-ary-bin-tail}
    Let $X \sim \mathrm{Bin}(n,p)$ with $0 < p < 1$, and let $p < p' < 1$. Then
    \[
        \Pr[X \geq p' n] \leq e^{-nD(p'\|p)}.
    \]
\end{proposition}

As a very simple corollary, we note that if we define for $q \in \N$
\[
    D_q(x\|y)=x\log_q\frac xy + (1-x)\log_q \frac{1-x}{1-y},
\]
then the following holds:
\begin{corollary}\label{cor:q-ary-bin-tail}
    Let $X \sim \mathrm{Bin}(n,p)$ with $0 < p < 1$, and let $p < p' < 1$. Then,
    \[
        \Pr[X \geq p'n] \leq q^{-nD_q(p'\|p)}.
    \]
\end{corollary}

We will make use of the following properties of KL divergence. 

\begin{claim} \label{claim:div-properties}
    Suppose $1>p_1 > p_2 > p_3>0$. Then, 
    \[
        D_q(p_1\|p_3) > D_q(p_1\|p_2) > 0.
    \]
\end{claim}

\begin{claim} \label{lem:pinsker}
    Let $p \in (0,1)$ and $0\leq \Delta < p$. Then, 
    \[
        D(p\|p-\Delta) \geq 2\Delta^2.
    \]
\end{claim}

Additionally, we need to lower bound the probability that a binomial random variable hits exactly its expectation. 

\begin{proposition} \label{prop:prob-bin-hits-expectation}
    Let $X \sim \mathrm{Bin}(n,p)$ with $0 < p < 1$ such that $pn$ is an integer.
    Then,
    \begin{equation*}
        \Pr[X=pn]\geq \frac{1}{\sqrt{2n}}.
    \end{equation*}
\end{proposition}
\begin{proof}
    First, note that
    \begin{equation*}
        \Pr[X=pn] = \binom{n}{pn} p^{pn} (1-p)^{(1-p)n} = \binom{n}{pn} \cdot 2^{-n h_2(p)},
    \end{equation*}
    The desired statement follows by combining this observation with a bound of Wozencraft and Reiffen~\cite{WR61} (see~\cite[Lemma 17.5.1]{CT12}), which in particular gives that
    \begin{equation*}
        \binom{n}{pn} \cdot 2^{-n h_2(p)} \geq \frac{1}{\sqrt{8p(1-p)n}}\geq \frac{1}{\sqrt{2n}}. \qedhere
    \end{equation*}
\end{proof}

\subsection{Increasing Chains}

The following definition of an increasing chain was first introduced by Guruswami, H{\aa}stad, and Kopparty~\cite{GHK11}.
\begin{definition}[$c$-increasing chain]
    A sequence of vectors $v_1,\dots,v_{d}\subseteq \F_q$ is said to be a \emph{$c$-increasing chain of length $d$} if for all $j\in [d]$ we have
    \begin{equation*}
        \left|\supp(v_j)\setminus \left(\bigcup_{i=1}^{j-1}\supp(v_i)\right)\right|\geq c.
    \end{equation*}
\end{definition}

We require the following lemma on the existence of appropriately long increasing chains in any subset $S\subseteq \F_q$.
\begin{lemma}[\protect{\cite[Lemma 6.3]{GHK11}}]\label{lem:SS}
    For every prime power $q$, and all positive integers $c,\ell$ and $L \le q^\ell$, 
    the following holds. For every $S \subseteq \F_q^\ell$ with $|S|=L$, there is 
    $w \in \F_q^{\ell}$ such that $S+w$ has a $c$-increasing chain of length at least $\frac{1}{c}\log_{q}\frac{L}{2} - (1-\frac{1}{c})\log_{q}((q-1)\ell)$.
\end{lemma}

\subsection{Mixing Sets}

In our analysis we need to understand the probability that the sum of two independent uniformly random samples $X$ and $X'$ from a set $T\subseteq \F_q$ lands in a shifted set $T+\gamma$, for an arbitrary shift $\gamma\in\F_q$ (and in fact a more general question of that form). We begin with the necessary definitions. 

\begin{definition}[mixing over $\F_q$]
For a prime power $q$, and $\delta \ge 0$,
we say that $T \subseteq \F_q$ is $\delta$-mixing, if for any $\alpha,\beta,\gamma \in \F_q$, where $\alpha$ and $\beta$ are nonzero, it holds that
\[
\Pr_{X,X'}[\alpha X + \beta X' \in T+\gamma] \le q^{-\delta},
\]
where $X,X' \sim T$ are independent and uniformly distributed over $T$, and $T+\gamma = \set{t+\gamma : t \in T}$.
\end{definition}

\begin{definition}[mixing over $\F_q^n$]
For $n \in \mathbb{N}$,  a prime power $q$, and $\delta \ge 0$, we say that $T \subseteq \F^n_q$ is $\delta$-mixing, if for any nonzero $\alpha,\beta \in \F_q$, and any $z \in \F_q^n$, it holds that
\[
\Pr_{X,X'}[\alpha X + \beta X' \in T+z] \le q^{-\delta n},
\]
where $X,X' \sim T$ are independent and uniformly distributed over $T$.
\end{definition}

\begin{remark}
\em
Note that a set mixes nontrivially when $\delta > 0$, and moreover, we will want $\delta > 0$ to not depend on $n$. However, in the case of list recovery from erasures, where for some $i$-s, $S_i = \F_q$, the case of $\delta = 0$ will be useful towards bounding the \emph{expected} mixing of $S_1,\ldots,S_n$.
\end{remark}

The following connection then follows easily.
\begin{remark}\label{remark:mixing}
\em
Suppose that $T_1,\ldots,T_n \subseteq \F_q$, and each $T_i$ is $\delta_i$-mixing. Then, $T = T_1 \times \cdots \times T_n$ is $\delta$-mixing, for $\delta = \E_{i\sim [n]}[\delta_i]$. 
\end{remark}

\section{List-Recovery from Erasures over Prime Fields}\label{sec:erasures}

In this section we establish the list-recoverability of random linear codes over prime fields.
To achieve this, we must first understand the mixing properties of worst-case subsets of prime fields.

\subsection{Worst-Case Mixing of Subsets of Prime Fields}\label{sec:sumset}

Towards understanding mixing of subsets of prime fields, we leverage a general result of Lev~\cite{Lev01} which characterizes worst-case $T$-s, when $q$ is prime.
Before we introduce it, we set up some relevant notation.

For a set $T\subseteq \F_q$ we let $\widetilde{T}$ denote a ``centered interval'' of length $|T|$.
More precisely, $\widetilde{T}$ is the \emph{$\delta$-centered interval} associated with $T$ if $\widetilde{T}=[-\alpha,\alpha+\delta]\subseteq \F_q$ with $\alpha\in[0,\frac{q-1}{2}]$ and $\delta\in\{-1,0,1\}$ satisfying
$|\widetilde{T}|=2\alpha+1+\delta=|T|$.
Note that when $|T|$ is odd there is a unique centered interval $\widetilde T$ (because $\delta=0$ necessarily), but when $|T|$ is even there are two centered intervals, corresponding to $\delta=\pm 1$.

\begin{lemma}[\protect{\cite[Theorem 1]{Lev01}, adapted}]\label{lem:lev}
    Let $q\geq 3$ be prime and $A_1,\dots,A_k\subseteq \F_q$ be arbitrary sets with $\widetilde A_1,\dots,\widetilde A_k$ the associated $\delta_i$-centered intervals.
    Then, if $|\delta_1+\dots+\delta_k|\leq 1$, for any set $B\subseteq \F_q$ and some associated $\delta$-centered interval $\widetilde B$, we have
    \begin{equation*}
        \Pr_{X_i\sim A_i}[X_1+\cdots+X_k\in B] \leq \Pr_{\widetilde X_i\sim \widetilde A_i}[\widetilde X_1+\cdots+\widetilde X_k\in \widetilde B].
    \end{equation*}
\end{lemma}

We can use \cref{lem:lev} to prove the following.
\begin{lemma}\label{lem:mixing}
    Fix a prime $q\geq 3$. Let $T_1,T_2,T_3\subseteq \F_q$ be arbitrary sets of size $\ell>0$.
    Then,
    \begin{equation*}
    \Pr_{X_1\sim T_1,X_2\sim T_2}[X_1+X_2\in T_3] \leq \begin{cases}
        \frac{3}{4}+ \frac{1}{4\ell^2} +\frac{\max(0,3\ell-2q-1)\cdot(3\ell-2q+1)}{4\ell^2},& \textrm{ if $\ell$ is odd,}\\
        \frac{3}{4} + \frac{\max(0,3\ell-2q)^2}{4\ell^2},& \textrm{ if $\ell$ is even,}
    \end{cases}
    \end{equation*}
    and this is tight for all $\ell$.
    In particular:
    \begin{enumerate}
        \item \label{item:23}
        When $\ell\leq 2q/3$ we have
    \begin{equation*}
       \Pr_{X_1\sim T_1,X_2\sim T_2}[X_1+X_2\in T_3]\leq \begin{cases}
           \frac{3}{4}+\frac{1}{4\ell^2}, &\textrm{ if $\ell$ is odd,}\\
           \frac{3}{4},& \textrm{ if $\ell$ is even.}
       \end{cases} 
    \end{equation*}

    \item \label{item:above23}
    When $2q/3<\ell\leq q-1$ we have
    \begin{equation*}
        \Pr_{X_1\sim T_1,X_2\sim T_2}[X_1+X_2\in T_3]\leq \frac{q^2-3\ell(q-\ell)}{\ell^2}.
    \end{equation*}
    
    \end{enumerate}
\end{lemma}
\begin{proof}
    By \cref{lem:lev}, we have
    \begin{equation*}
        \Pr_{X_1\sim T_1,X_2\sim T_2}[X_1+X_2\in T_3] \leq \Pr_{\widetilde X_1 \sim \widetilde T_1,\widetilde X_2\sim \widetilde T_2}[\widetilde X_1+\widetilde X_2\in \widetilde T_3],
    \end{equation*}
    where $\widetilde T_1, \widetilde T_2, \widetilde T_3$ are (possibly different) centered intervals associated with $T_1,T_2,T_3$, respectively.
    We consider two cases depending on the parity of $\ell$.

    If $\ell$ is odd, then $\widetilde T_1=\widetilde T_2=\widetilde T_3=[-\alpha,\alpha]$ with $\alpha=\frac{\ell-1}{2}$.
    Therefore,
    \begin{align*}
        \Pr_{X_1\sim T_1,X_2\sim T_2}[X_1+X_2\in T_3] &\leq \Pr_{\widetilde X_1,\widetilde X_2\sim [-\alpha,\alpha]}[\widetilde X_1+\widetilde X_2\in [-\alpha,\alpha]]\\
        &= \sum_{i=-\alpha}^\alpha \sum_{j=-\alpha}^\alpha \frac{1}{\ell^2} \cdot 1_{\{i+j\Mod q \in [-\alpha,\alpha]\}},
    \end{align*}
    where we interpret the mod $q$ operation as mapping to $[-\frac{q-1}{2},\frac{q-1}{2}]$.
    The term inside the double sum is $1$ exactly when one of the following three disjoint cases occurs:
    \begin{itemize}
        \item $i+j\in[-\alpha,\alpha]$:
        By symmetry we can focus on $i\geq 0$, in which case this is equivalent to $j\in[-\alpha,\alpha-i]$.
        Therefore, the probability that this event occurs is (recall that $\ell=2\alpha+1$)
        \begin{equation*}
            \frac{1}{\ell}+\frac{2}{\ell^2}\sum_{i=1}^\alpha \left((\alpha-i)-(-\alpha)+1\right)=\frac{1}{\ell}+\frac{2}{\ell^2}\sum_{i=1}^\alpha (\ell-i) = \frac{3}{4}+\frac{1}{4\ell^2}.
        \end{equation*}
        The factor of $2$ comes from grouping together $\pm i$ for $i>0$ by symmetry. The leftmost $\frac{1}{\ell}$ term corresponds to $i=0$. 

        \item $i+j\in[q-\alpha,q]$: This is equivalent to $j\in[q-\alpha-i,\alpha]$, because $i+j\in[q-\alpha,q]$ is equivalent to $j\in [q-\alpha-i,q-i]$. We also know that $j\in [-\alpha,\alpha]$. Since $q-i\geq \alpha$ for all $i\leq \alpha$ (because $\alpha\leq \frac{q-1}{2}$), the intersection of the two intervals is $[q-\alpha-i,\alpha]$.
        Note that the interval $[q-\alpha-i,\alpha]$ is non-empty if and only if $i\geq q-2\alpha = q-\ell$.
        Therefore, the probability of this case occurring is
        \begin{equation*}
            \frac{1}{\ell^2} \sum_{i=q-\ell}^\alpha \left(\alpha-(q-\alpha-i)+1\right) = \frac{1}{\ell^2} \sum_{i=q-\ell}^\alpha (\ell- q+i).
        \end{equation*}
        This sum is nonzero only when $\alpha> q-\ell$, which is equivalent to $3\ell-2q-1>0$.
        Therefore, this expression simplifies to
        \begin{equation*}
            \frac{\max(0,3\ell-2q-1)\cdot (3\ell-2q+1)}{8\ell^2}.
        \end{equation*}

        \item $i+j\in[-q,-q+\alpha]$: By symmetry, the probability that this case occurs is equal to the previous case.
    \end{itemize}
    Summing the three probabilities yields the desired result for odd $\ell$-s.

    If $\ell$ is even, we can choose the centered intervals $\widetilde T_1=[-\alpha,\alpha-1]$ and $\widetilde T_2=[-\alpha+1,\alpha]$ for $\alpha=\ell/2$.
    The $\delta_3$-centered interval $\widetilde T_3$ associated with $T_3$ also must have $\delta_3=\pm 1$. Without loss of generality we assume that $\widetilde T_3=\widetilde T_2 = [-\alpha+1,\alpha]$ (otherwise, flip $\alpha$ to $-\alpha$).
    Now,
    \begin{align*}
        \Pr_{X_1\sim T_1,X_2\sim T_2}[X_1+X_2\in T_3]  &\leq \Pr_{\widetilde X_1 \sim [-\alpha,\alpha-1],\widetilde X_2 \sim [-\alpha+1,\alpha]}[\widetilde X_1+\widetilde X_2\in [-\alpha+1,\alpha]]
        \\
        &= \frac{1}{\ell^2}\sum_{i=-\alpha}^{\alpha-1}\sum_{j=-\alpha+1}^\alpha 1_{\{i+j\Mod q \in [-\alpha+1,\alpha]\}},
    \end{align*}
    where, as before, we intepret the mod $q$ operation as mapping to $[-\frac{q-1}{2},\frac{q-1}{2}]$.
    The term inside the double sum is $1$ exactly when one of the following three disjoint cases occurs:
    \begin{itemize}
        \item $i+j\in[-\alpha+1,\alpha]$: This is equivalent to $j\in[-\alpha+1,\alpha-i]$ if $i\geq 0$ and to $j\in [-\alpha+1-i,\alpha]$ if $i<0$.
        Therefore, the probability of this case occurring is (recall that $\ell=2\alpha$)
        \begin{align*}
            &\frac{1}{\ell^2} \sum_{i=0}^{\alpha-1} (\alpha-i -(-\alpha+1)+1) + \frac{1}{\ell^2} \sum_{i=-\alpha}^{-1}(\alpha-(-\alpha+1-i)+1) \\
            &= \frac{1}{\ell^2} \sum_{i=0}^{\alpha-1} (\ell -i) + \frac{1}{\ell^2} \sum_{i=-\alpha}^{-1} (\ell+i)= \frac{3}{4}.
        \end{align*}

        \item $i+j\in[q-\alpha+1,q]$: This is equivalent to $j\in [q-\alpha+1-i,\alpha]$.
        Therefore, the probability of this case occurring is
        \begin{equation*}
            \frac{1}{\ell^2} \sum_{i=-\alpha}^{\alpha-1} \max(0,(\alpha-(q-\alpha+1-i)+1)) = \frac{1}{\ell^2} \sum_{i=-\alpha}^{\alpha-1} \max(0,\ell-q+i) = \frac{1}{\ell^2} \sum_{i=q-\ell}^{\alpha-1} (\ell-q+i).
        \end{equation*}
        The rightmost sum is positive only if $\alpha-1\geq q-\ell$, which is equivalent to $3\alpha-q>0$.
        Therefore, we can simplify that right-hand side as 
        \begin{equation*}
            \frac{\max(0,3\alpha-q)\cdot (3\alpha-q-1)}{2\ell^2} = \frac{\max(0,3\ell-2q)\cdot (3\ell-2q-2)}{8\ell^2}.
        \end{equation*}

        \item $i+j\in[-q,-q+\alpha]$: This is equivalent to $j\in [-\alpha+1,-q+\alpha-i]$.
        Therefore, the probability of this case occurring is
        \begin{equation*}
            \frac{1}{\ell^2} \sum_{i=-\alpha}^{\alpha-1} \max(0,-q+\alpha-i-(-\alpha+1)+1) = \frac{1}{\ell^2} \sum_{i=-\alpha}^{\alpha-1} \max(0,\ell -q-i) = \frac{1}{\ell^2} \sum_{i=-\alpha}^{\ell-q} (\ell -q-i).
        \end{equation*}
        The rightmost sum is positive only if $3\alpha-q>0$, and so we can simplify the right-hand side as
        \begin{equation*}
            \frac{\max(0,3\alpha-q)\cdot (3\alpha-q+1)}{2\ell^2} =\frac{\max(0,3\ell-2q)\cdot (3\ell-2q+2)}{8\ell^2}.
        \end{equation*}
    \end{itemize}
    Summing the probabilities for the three cases yields the desired result for even $\ell$-s.

    Regarding the ``In particular'' part, \cref{item:23} follows directly by noticing that all maximums are $0$ when $\ell\leq 2q/3$.
    \cref{item:above23} uses the fact that when $\ell>2q/3$ both upper bounds (for even and odd $\ell$) simplify to $\frac{q^2-3\ell(q-\ell)}{\ell^2}$.
\end{proof}
We can then record the following corollary. 

\begin{corollary}\label{cor:prime-mixing}
For a prime $q \ge 3$, any set $T \subseteq \F_q$ of size $\ell \le q-1$,
\begin{itemize}
    \item If $2 \le \ell \le 2q/3$, $T$ is $\delta$-mixing for $\delta \ge \log_{q}(16/13)$.
    \item Otherwise, $T$ is $\delta$-mixing for $\delta \ge \log_{q}\left( \frac{\ell^2}{q^2 - 3\ell(q-\ell)} \right)$.
\end{itemize}
\end{corollary}

\subsection{List-Recovery from Erasures over Prime Fields via Mixing}

In this section we adapt the technique in \cite{GHK11}, together with the worst-case mixing result from \cref{cor:prime-mixing}, to establish list recovery from erasures over prime fields.

\begin{lemma}\label{lem:main-erasures}
Let $T \subseteq \F_q^n$ be $\delta$-mixing. For $b \in \mathbb{N}$ and any $a > 0$ satisfying $n \ge q^{\frac{8a}{\delta}}$, the following holds. Let $X^{(1)},\ldots,X^{(b)}$ be sampled independently and uniformly at random from $T$. Then, we have that
\[
\Pr\left[ \left| \Sp\left(X^{(1)},\ldots,X^{(b)}\right) \cap T\right| > C \cdot b \right] \le q^{-a n},
\]
where $C = C(q,\delta,a) = q^{\frac{8a}{\delta}}$.
In particular, when $T = T_1 \times \cdots \times T_n$, and each $T_i$ is $\delta_i$-mixing, we get the same result as above for $\delta = \E_{i}[\delta_i]$.
\end{lemma}
\begin{proof}
Let $E$ denote the bad event that we want to bound, namely $|\Sp(X^{(1)},\ldots,X^{(b)}) \cap T| > A$ for $A = b \cdot q^{\frac{8a}{\delta}}$. Note that $E$ implies that there exists some set $S \subseteq \F_q^b$, $|S|=A+1$, such that $X_v \triangleq \sum_{i \in [b]}v_i X^{(i)} \in T$ for all $v\in S$.\footnote{Notice that if the $X_i$-s are not linearly independent, this can only decrease the probability that the intersection is large, so we can concentrate on the case that distinct $v$-s give rise to distinct $X_v$-s.} Hence, it suffices to bound the probability that such a set $S$ exists.

Fix some $S \subseteq \F_q^b$ of size $A+1$.
Applying \cref{lem:SS} with $c=2$ (and note that we can assume that $A +1 \le q^b$), we know there exists $w \in \F_q^b$ such that $S+w$ has a $2$-increasing chain of length
$d = \frac{1}{2}\log_{q}\frac{A+1}{2} - \frac{1}{2}\log_{q}((q-1)b)$. That is, we have $v^{(1)},\ldots,v^{(d)} \in S$ such that for all $j \in [d]$,
\[
\left| \supp(v^{(j)}) \setminus \left(\bigcup_{i = 1}^{j-1}\supp(v^{(i)})\right) \right| \ge 2.
\]
Now, we can bound
\begin{align}
\Pr[\forall v \in S, X_v \in T] &\le \Pr[\forall j \in [d], X_{v^{(j)}} \in T] \nonumber \\
&= \Pr[\forall j \in [d], X_{v^{(j)}}+X_w \in T + X_w] \nonumber \\
&= \Pr[\forall j \in [d], X_{v^{(j)}+w} \in T + X_w]. \label{eq:1}
\end{align}
Next, we bound \cref{eq:1} by 
\begin{equation}\label{eq:3}
\sum_{y\in \F_q^n}\Pr[\forall j \in [d], X_{v^{(j)}+w} \in T + y].
\end{equation}

Towards bounding each term in the sum, observe that the increasing chain property tells us that for each $j \in [d]$, we can write $X_{v^{(j)}+w}=Y_{\mathsf{past}}^{(j)}+Y_{\mathsf{new}}^{(j)}$, where $Y^{(j)}_{\mathsf{past}}$ contains $X_k$-s that participated in $\set{X_{v^{(i)}+w}}_{i < j}$,
whereas $Y_{\mathsf{new}}^{(j)}$ contains two new $X_{k}$-s. Now,
\begin{equation}\label{eq:2}
\Pr[X_{v^{(j)}+w} \in T+y ~|~ \forall i \in [j-1], X_{v^{(i)}+w} \in T+y] = 
\E_{z \sim Y_{\mathsf{past}}^{(j)}}\left[ \mathbf{1}(z) \cdot \Pr[Y_{\mathsf{new}}^{(j)} \in T+y_z] \right],
\end{equation}
where $\mathbf{1}(z)$ is an indicator for whether past $X_k$-s landed in $T+y$, and $y_z$ is a fixed string that depends on the fixing of $Y^{(j)}_{\mathsf{past}}$. Assume for simplicity that $Y^{(j)}_{\mathsf{new}} = \alpha X^{(1)} + \beta X^{(2)}$, where $\alpha,\beta \in \F_q$ are nonzero. Then, 
using the fact that $T$ is $\delta$-mixing,
each summand of \cref{eq:3} can now be bounded by
\[
\prod_{j \in [d]}\Pr[X_{v^{(j)}+w} \in T+y ~|~ \forall i \in [j-1], X_{v^{(i)}+w} \in T+y] \le q^{-n\delta d},
\]
and summing over all $y$-s gives us
\[
\Pr[v \in S, X_v \in T] \le q^{n} \cdot q^{-n\delta d} = q^{-(\delta d - 1)n}.
\]
Union-bounding over all $S$-s, we get
\begin{equation}\label{eq:4}
\Pr[E] \le \binom{q^b}{A+1}q^{-(\delta d - 1)n} \le q^{b(A+1) - (\delta d -1)n}.
\end{equation}
First, note that we set parameters so that $d \ge \frac{2a+1}{\delta}$. Indeed, we can set $d = \left\lfloor \frac{1}{2}\log_{q}\left( \frac{A+1}{2b(q-1)} \right)\right\rfloor$, and then need $A$ to be at most, say, $4q \cdot q^{\frac{4a}{\delta}} \le q^{\frac{8a}{\delta}}$. Under this choice of $A$, it also holds that $b(A+1) \le \frac{\delta d-1}{2}n$, since $n$ is large enough. Overall, \cref{eq:4} gives $\Pr[E] \le q^{-\frac{\delta d - 1}{2} n} \le q^{-an}$, as desired. The ``In particular'' part simply follows from \cref{remark:mixing}.
\end{proof}

We are now ready to prove our list recovery result.
\begin{theorem}[list recovery with erasures]\label{thm:erasures}
For any $n \in \mathbb{N}$, a prime $q$, an integer $\ell \le q-1$, and $\alpha,\eps \in (0,1)$, the following holds. With probability at least $1-q^{-n}$, a random linear code $\mathcal{C}\subseteq \F_q^n$ of rate \[R = 1-\alpha - (1-\alpha)\log_{q}\ell-\eps\] is $(\alpha,\ell,L)$-list-recoverable from erasures, with
\[
L = C_{q,\ell,\alpha} \cdot \frac{1}{\eps},
\]
provided that $n \ge L$. In particular, there exists a universal constant $C$ such that:
\begin{itemize}
    \item When $\ell \le \frac{2}{3}q$, we can take $C_{q,\ell,\alpha} \le q^{C \log q \cdot ((1-\alpha)\ell+1)} \triangleq C_{q,\ell,\alpha}^{(0)}$, and,
    \item When $\ell = (1-\gamma) q$ for some $\gamma \in (0,1/3)$, we can take $C_{q,\ell,\alpha} \le \left( C_{q,\ell,\alpha}^{(0)} \right)^{\frac{(1-\gamma)^2}{1-3\gamma(1-\gamma)}}$.
\end{itemize}
\end{theorem}
\begin{proof}
Fix some $T_1,\ldots,T_n \subseteq \F_q$, and some set $\mathbf{B} \subseteq [n]$ of density $\alpha$, such that for all $i \notin \mathbf{B}$, $|T_i| \le \ell$, and otherwise $T_i = \F_q$. Denote $T = T_1 \times \ldots \times T_n$, and record that there are
\[
\mathcal{T} \le \binom{n}{\alpha n}\binom{q}{\ell}^{(1-\alpha)n} \le 2^{h_2(\alpha)n} \cdot \left( \frac{eq}{\ell}\right)^{(1-\alpha)n \ell}
= q^{\frac{h_2(\alpha)}{\log q}n + (1-\alpha)n\ell }
\]
options to choose such a $T$.

For each $b \in [\log_{q}(L+1),L+1]$, let $\mathcal{F}_b$ be the set of all $(v_1,\ldots,v_b) \in T^{b}$ such that the $v_i$-s are linearly independent, and $|\Sp(\set{v_1,\ldots,v_b}) \cap T| > L$. Further, let $\mathcal{F} = \bigcup_{b}\mathcal{F}_b$. Given a tuple $\bm{v} = (v_1,\ldots,v_b)$, we denote by $\set{\bm{v}}$ the (unordered) set $\set{v_1,\ldots,v_b}$.

Note that if $|T \cap \mathcal{C}| > L$, there must exist some $\bm{v} \in \mathcal{F}$ for which $\set{\bm{v}} \subseteq T \cap \mathcal{C}$, by taking a maximal linearly independent subset of $T \cap \mathcal{C}$. Now,
for some fixed $\bm{v} \in \mathcal{F}_b$,
\[
\Pr_{\mathcal{C}}[\bm{v} \subseteq T \cap \mathcal{C}] \le \Pr_{\mathcal{C}}[\bm{v} \subseteq \mathcal{C}] \le q^{-(1-R)nb},
\]
where we used \cref{prop:prob-contained-in-rlc}. Overall, to bound the probability $p$ that the code is not 
list-recoverable, we need to union-bound over 
the choice of $T$, and over the $\set{\bm{v}}$-s. 
Concretely,
\begin{equation}\label{eq:5}
p \le \mathcal{T} \cdot \sum_{b \in [\log_{q}(L+1),L+1]}\sum_{\bm{v} \in \mathcal{F}_b}q^{-(1-R)nb}.
\end{equation}
We break the above sum with respect to a threshold $b^{\star} = \left\lceil c \cdot \frac{\ell}{\eps}\right\rceil$, for $c$ to be determined soon. When $b < b^{\star}$, we apply \cref{lem:main-erasures} as follows. The parameter $a$ will be determined later on. To determined $\delta$, note that whenever $T_i = \F_q$, $T_i$ is trivially $0$-mixing, whereas for $|T_i|=\ell$, $T_i$ is $\delta_i$-mixing for either
$
\delta_0 \ge \log_{q}(16/13),
$
or
$
\delta_0 \ge \log_{q}\left( \frac{\ell^2}{q^2 - 3\ell(q-\ell)} \right),
$
following \cref{cor:prime-mixing}. Thus,
$\delta \ge (1-\alpha)\delta_0$, and we will set $a$ and $c$ so that $L \ge q^{\frac{8a}{\delta}} \cdot \frac{c\ell}{\eps}$. 
Using our terminology, \cref{lem:main-erasures} tells us that 
\[
\Pr_{\bm{v} \sim T^b}\left[ \left| \Sp\left(\set{\bm{v}}\right) \cap T\right| > L \right] \le q^{-a n},
\]
so in particular, $|\mathcal{F}_b| \le |T|^b \cdot q^{-an}$,
and we can bound
\begin{align*}
p_1 &\triangleq \mathcal{T} \sum_{b= \log_{q}(L+1)}^{b^{\star}}\sum_{\bm{v}\in\mathcal{F}_b}q^{-(1-R)nb} \le \mathcal{T} \cdot b^{\star} \cdot q^{-an} \cdot \sum_{b=\log_{q}(L+1)}^{b^{\star}} |T|^b \cdot q^{-(1-R)nb} \\
&\le \mathcal{T} \cdot b^{\star}q^{-an}\sum_{b=\log_{q}(L+1)}^{b^{\star}}q^{\alpha nb}\ell^{(1-\alpha)nb}q^{-(1-R)nb} = \mathcal{T} \cdot  
b^{\star}q^{-an}\sum_{b=\log_{q}(L+1)}^{b^{\star}}q^{-\left( 1-R-\alpha-(1-\alpha)\log_{q}\ell \right)nb} \\
&\le \mathcal{T} \cdot  (b^{\star})^2 q^{-an}q^{-\eps \cdot n \cdot \log_{q}(L+1)} \le 
q^{\frac{h_2(\alpha)}{\log q}n + (1-\alpha)n\ell  + 2\log b^{\star}-an- \eps n \cdot \log_{q}(L+1)} \\ &\le q^{(1+(1-\alpha)\ell - a - \eps\log_{q}(L+1))n},
\end{align*}
where we used the fact that $\frac{h_2(\alpha)}{\log q}+\frac{2\log b^{\star}}{n} \le 1$. Now, $L+1 \ge q$, so it suffices to choose $a \ge (1-\alpha)\ell + 3$ to get that $p_1 \le q^{-2n}$.

Moving on to $p_2$, we have
\[
p_2 \triangleq \mathcal{T} \sum_{b= b^{\star}+1}^{L+1}\sum_{\bm{v}\in\mathcal{F}_b}q^{-(1-R)nb} \le \mathcal{T} \cdot \sum_{b=b^{\star}+1}^{L+1}|T|^b q^{-(1-R)nb} \le 
\mathcal{T} \cdot \sum_{b=b^{\star}+1}^{L+1}q^{-\left( 1-R-\alpha-(1-\alpha)\log_{q}\ell \right)nb},
\]
where we used the trivial bound on $|\mathcal{F}_b|$. Now, 
\[
p_2 \le \mathcal{T} \cdot (L+1) \cdot q^{-\eps n b^{\star}} \le q^{\frac{h_2(\alpha)}{\log q}n+(1-\alpha)n\ell+\log_{q}(L+1)-c n\ell} \le q^{(2+(1-\alpha-c)\ell)n}.
\]
Thus, choosing $c \ge 1-\alpha+\frac{4}{\ell}$, we get that $p_2 \le q^{-2n}$, and recalling \cref{eq:5}, we get that overall, $p \le q^{-n}$. Finally, recall that the bound on $p$ was provided that $a$ and $c$ are such that 
$L \ge q^{\frac{8a}{\delta}} \cdot \frac{c\ell}{\eps}$. Indeed, if we set $a = (1-\alpha)\ell+3$
and $c = 5$, this holds by choosing the constant $C$ in the theorem's statement appropriately. 
\end{proof}

\section{List-Recovery from Errors over Arbitrary Fields}

We now turn to the case of list-recovery from errors. Unlike in the case of list-recovery from erasures, we will make no assumptions on the underlying field. We firstly show that list-recovery balls are non-trivially mixing. We subsequently sketch how, again using \Cref{lem:mixing}, we can conclude the desired list-recovery result.

\subsection{Mixture of List-Recovery Balls}

In this section we make use of the \emph{convolution} of two functions $f,g\colon \F_q \to \R$, defined as 
\[
    (f*g)(x)=\mathop{\E}_{y \in \F_q}[f(y)g(x-y)] \ .
\]
We recall that it is \emph{distributive}, i.e., $f*(g+h) = f*g+f*h$ for any functions $f,g,h \colon \F_q \to \R$. We also note the following equality.
\begin{proposition} \label{prop:prob-to-convolution}
    Let $q$ be a prime power. Suppose $A,B \subseteq \F_q$ and $X,Y$ are sampled independently and uniformly from $A,B$, respectively. Then, for any $z \in \F_q$,
    \[
        \Pr[X+Y=z] = \frac{q}{|A|\cdot|B|}1_A*1_B(z).
    \]
\end{proposition}

\begin{proof}
    We have
    \begin{align*}
        \Pr[X+Y=z] &= \sum_{x \in A}\Pr[X=x]\Pr[X+Y=z|X=x] = \sum_{x \in A}\frac{1}{|A|}\cdot \Pr[Y=z-x]\\
        &= \frac{1}{|A|}\sum_{x \in \F_q}1_A(x) \cdot 1_B(z-x)\frac{1}{|B |} = \frac{q}{|A|\cdot |B|}\mathop{\E}_{x \in \F_q}[1_A(x)\cdot 1_B(z-x)] \\
        &= \frac{q}{|A|\cdot |B|} 1_A*1_B(z). \qedhere
    \end{align*}
\end{proof}

\begin{lemma}\label{lem:mixing-errors-helper}
    
Suppose $q$ is a prime power, $1 \leq \ell \leq q$ is an integer and $\omega_1,\omega_2 \in (0,1-\ell/q)$. Let $A,B,T \subseteq \F_q$ be of size $\ell$, let $X_1 \sim \Unif(A)$, $X_2 \sim \Unif(B)$, $Y_1 \sim \Unif(\F_q\setminus A)$ and $Y_2 \sim \Unif(\F_q\setminus B)$ all independent. Then, 
    \begin{align}
        (1-\omega_1)&(1-\omega_2)\Pr[X_1+X_2 \in T] + \omega_1(1-\omega_2)\Pr[X_1 +Y_2 \in T] \\
        &\qquad+\omega_2(1-\omega_1)\Pr[Y_1+X_2 \in T] \nonumber+ \omega_1\omega_2\Pr[Y_1+Y_2 \in T] \label{eq:lem-LHS-mixingness}\\
        &\leq (1-\omega_1)(1-\omega_2) + \omega_1\omega_2 \cdot \frac{\ell}{q-\ell} \ .
    \end{align}
\end{lemma}

\begin{proof}
     Using \Cref{prop:prob-to-convolution} and the observation that $1_{\F_q\setminus A}  = 1 - 1_A$ (and similarly for $B$), for any $z \in T$ we find
     \begin{align*}
        \Pr[X_1+X_2 =z] &= \frac{q}{\ell^2}\cdot 1_A * 1_B(z) \ , \\
        \Pr[X_1+Y_2 =z] &= \frac{q}{\ell(q-\ell)}\cdot 1_A * (1-1_B)(z) = \frac{q}{\ell(q-\ell)}\cdot\left( \frac{\ell}{q}- 1_A*1_B(z)\right) \ ,\\
        \Pr[X_2+Y_1=z] &= \frac{q}{\ell(q-\ell)}\cdot (1-1_A) * 1_B(z) = \frac{q}{\ell(q-\ell)}\cdot\left( \frac{\ell}{q}- 1_A*1_B(z)\right) \\ 
        \Pr[Y_1+Y_2 =z] &= \frac{q}{(q-\ell)^2} \cdot (1-1_A)*(1-1_B)(z) \\
        &= \frac{q}{(q-\ell)^2} \cdot \left(1-2\cdot \frac\ell q + 1_A*1_B(z)\right) \ ,
    \end{align*}
    where we also observed that $1*1_A(z) = \E_{x \in \F_q}[1\cdot 1_A(z-x)] = \frac{\ell}{q}$, and similarly $1*1_B(z) = \frac{\ell}{q}$. Thus, \cref{eq:lem-LHS-mixingness} becomes
    \begin{align}
        &\sum_{z \in T}\left[(1-\omega_1)(1-\omega_2)\cdot \frac{q}{\ell^2}\cdot 1_A*1_B(z) + (\omega_1(1-\omega_2)+\omega_2(1-\omega_1))\cdot \frac{q}{\ell(q-\ell)}\cdot\left( \frac{\ell}{q}- 1_A*1_B(z)\right) \right.\nonumber \\
        &\qquad\left.+ \omega_1\omega_2 \cdot \frac{q}{(q-\ell)^2} \cdot \left(1-2\cdot\frac\ell q + 1_A*1_B(z)\right)\right]\nonumber \\
        &= \sum_{z \in T}\left((1-\omega_1)(1-\omega_2)\cdot \frac{q}{\ell^2} - (\omega_1(1-\omega_2)+\omega_2(1-\omega_1))\frac{q}{\ell(q-\ell)} + \rho^2\cdot \frac{q}{(q-\ell)^2}\right)\cdot 1_A*1_B(z) \nonumber \\ 
        &\qquad + \ell\cdot\left[(\omega_1(1-\omega_2)+\omega_2(1-\omega_1))\frac{q}{\ell(q-\ell)}\cdot\frac{\ell}q + \omega_1\omega_2\frac{q}{(q-\ell)^2}\left(1 - 2\cdot \frac{\ell}{q}\right)\right] \nonumber\\ 
        &= \left((1-\omega_1)(1-\omega_2)\cdot \frac{q}{\ell^2} - (\omega_1(1-\omega_2)+\omega_2(1-\omega_1))\frac{q}{\ell(q-\ell)} \right.\nonumber \\
        &\qquad\qquad\qquad + \left.\omega_1\omega_2\cdot \frac{q}{(q-\ell)^2}\right) \sum_{z \in T}1_A*1_B(z) \label{eq:main-sum} \\
        &\qquad+ \frac{(\omega_1(1-\omega_2)+\omega_2(1-\omega_1))\ell}{q-\ell} + \omega_1\omega_2\cdot \frac{q\ell-2\ell^2}{(q-\ell)^2} \label{eq:rest} \ .
    \end{align}
     We first consider the term \cref{eq:main-sum}. Let 
    \[
        c = (1-\omega_1)(1-\omega_2)\cdot \frac{q^2}{\ell^2} - (\omega_1(1-\omega_2)+\omega_2(1-\omega_1))\frac{q^2}{\ell(q-\ell)} + \omega_1\omega_2\cdot \frac{q^2}{(q-\ell)^2} \ .
    \]
    We claim that $c\geq 0$. Indeed, put $\alpha:=\frac{\ell}{q}$, so $1-\alpha = 1-\frac{\ell}{q} = \frac{q-\ell}{q}$ and note that 
    \[
        c = \frac{(1-\omega_1)(1-\omega_2)}{\alpha^2} -  \frac{\omega_1(1-\omega_2)}{\alpha\cdot (1-\alpha)} - \frac{\omega_2(1-\omega_1)}{\alpha\cdot (1-\alpha)} + \frac{\omega_1\omega_2}{1-\alpha^2} \ .
    \]
    Put $s = \frac{1-\omega_1}\alpha$, $t=\frac{1-\omega_2}\alpha$, $u = \frac{\omega_1}{1-\alpha}$ and $v = \frac{\omega_2}{1-\alpha}$, so then 
    \begin{align}
        c = st - sv - tu + uv = s(t-v) - u(t-v). \label{eq:c>=0}
    \end{align}
    Since $\omega_1<1-\alpha$, it follows that $u \leq 1 \leq s$, and similarly as $\omega_2 < 1-\alpha$ we have $v \leq 1 \leq t$. Thus, looking at \cref{eq:c>=0}, since $t-v\geq 0$ and $s \geq u$, it follows that $c \geq 0$, as claimed. 

    Thus, using the trivial upper bound $1_A*1_B(z) \leq \frac{\ell}q$ valid for all $z \in \F_q$, we find
    \[
        \eqref{eq:main-sum} = \frac{c}q\sum_{z \in T}1_A*1_B(z) \leq \frac{c}q \cdot \ell \cdot \frac{\ell}{q} = \frac{\ell^2}{q^2}\cdot c = \alpha^2 c\ .
    \]
    We can expand
    \begin{align*}
        \alpha^2\cdot c &= \alpha^2 \cdot\left(\frac{(1-\omega_1)(1-\omega_2)}{\alpha^2} -  \frac{\omega_1(1-\omega_2)}{\alpha\cdot (1-\alpha)} - \frac{\omega_2(1-\omega_1)}{\alpha\cdot (1-\alpha)} + \frac{\omega_1\omega_2}{1-\alpha^2}\right) \\
        &= (1-\omega_1)(1-\omega_2) - (\omega_1(1-\omega_2)+\omega_2(1-\omega_1))\cdot \frac{\alpha}{1-\alpha} + \omega_1\omega_2\left(\frac{\alpha^2}{(1-\alpha)^2}\right) \\ 
        &= (1-\omega_1)(1-\omega_2) - (\omega_1(1-\omega_2)+\omega_2(1-\omega_1))\frac{\ell}{q-\ell} + \omega_1\omega_2 \frac{\ell^2}{(q-\ell)^2} \ .
    \end{align*}
    Combining this again with \cref{eq:rest}, we obtain the following upper bound:
    \begin{align*}
    	(1-\omega_1)&(1-\omega_2) - (\omega_1(1-\omega_2)+\omega_2(1-\omega_1))\frac{\ell}{q-\ell} + \omega_1\omega_2 \frac{\ell^2}{(q-\ell)^2} \\
        &\qquad\qquad + (\omega_1(1-\omega_2)+\omega_2(1-\omega_1))\frac{\ell}{q-\ell} + \omega_1\omega_2\frac{q\ell-2\ell^2}{(q-\ell)^2} \\
    	&= (1-\omega_1)(1-\omega_2) + \omega_1\omega_2\left(\frac{q\ell}{(q-\ell)^2} - \frac{\ell^2}{q-\ell}\right)\\
        &= (1-\omega_1)(1-\omega_2) + \omega_1\omega_2\cdot \frac{\ell}{q-\ell}\left(\frac{q}{q-\ell}-\frac{\ell}{q-\ell}\right) \\
    	&= (1-\omega_1)(1-\omega_2) + \omega_1\omega_2\cdot \frac{\ell}{q-\ell}\left(\frac{q-\ell}{q-\ell}\right) = (1-\omega_1)(1-\omega_2) + \omega_1\omega_2\cdot  \frac{\ell}{q-\ell} \ ,
    \end{align*}
    as was to be shown. 
\end{proof}

In the sequel, we will need upper and lower bounds on the RHS of \Cref{lem:mixing-errors-helper} to $1-\rho$. The following lemma establishes the required bounds.

\begin{lemma} \label{lem:omega-ineq}
    Suppose $q$ is a prime power, $1 \leq \ell \leq q$ is an integer and $\rho \in (0,1-\ell/q)$. Suppose $0 \leq \eta \leq \rho$ and $\rho - \eta \leq \omega_1,\omega_2 \leq \rho$. Then, the following inequalities hold:
    \begin{align}
        1-\rho &> (1-\rho)^2 + \rho^2\cdot \frac{\ell}{q-\ell}, \label{eq:first-ineq-helper} \\
        (1-\rho+\eta)^2 &+ (\rho-\eta)^2 \cdot \frac{\ell}{q-\ell} > (1-\omega_1)(1-\omega_2) + \omega_1\omega_2 \cdot \frac{\ell}{q-\ell}. \label{eq:sec-ineq-helper}
    \end{align}
\end{lemma}

\begin{proof}
    For \cref{eq:first-ineq-helper}, observe that the quadratic
    \[
        (1-\rho) = (1-\rho)^2 + \rho^2\cdot \frac{\ell}{q-\ell},
    \]
    has solutions $\rho = 0$ and $\rho = 1-\ell/q$, from which it follows that 
    \begin{align}
        1-\rho > (1-\rho)^2 + \rho^2\cdot \frac{\ell}{q-\ell} ~~~ \forall~ 0 < \rho < 1-\frac{\ell}q,
    \end{align}
    as claimed. 

    We now consider \cref{eq:sec-ineq-helper}. Consider the function $g(\omega_1,\omega_2) = (1-\omega_1)(1-\omega_2) + \omega_1\omega_2 \cdot \frac{\ell}{q-\ell}$, defined for $(\omega_1,\omega_2) \in [\rho-\eta,\rho]^2$. To show~\cref{eq:sec-ineq-helper}, it suffices to show that $g$ is strictly maximized at the point $(\rho-\eta,\rho-\eta)$. If $\omega_1 > \rho-\eta$, we first argue $g(\rho-\eta,\omega_2) > g(\omega_1,\omega_2)$. Put $\eps_1 = \omega_1-(\rho-\eta)>0$ and observe that 
    \begin{align*}
        g&(\rho-\eta,\omega_2) - g(\omega_1,\omega_2) \\
        &= \left((1-(\omega_1-\eps_1))(1-\omega_2) + (\omega_1-\eps_1)\omega_2 \cdot \frac{\ell}{q-\ell}\right) - \left((1-\omega_1)(1-\omega_2) + \omega_1\omega_2 \cdot \frac{\ell}{q-\ell}\right)\\
        &= \eps_1(1-\omega_2) -\eps_1\omega_2\frac{\ell}{q-\ell} = \eps_1\left(1-\omega_2\left(1+\frac{\ell}{q-\ell}\right)\right).
    \end{align*}
    Since $\omega_2 <1-\frac{\ell}q$, we observe that
    \begin{align*}
        \omega_2\cdot \left(1+\frac{\ell}{q-\ell}\right) &< \left(1-\frac{\ell}q\right)\left(1+\frac{\ell}{q-\ell}\right) = 1 - \frac{\ell}q + \frac{\ell}{q-\ell} - \frac{\ell^2}{q(q-\ell)} \\
        &= 1 - \frac{\ell(q-\ell)-\ell q + \ell^2}{q(q-\ell)} = 1 - \frac{\ell q - \ell^2 -\ell q + \ell^2}{q(q-\ell)} = 1.
    \end{align*}
    Thus, 
    \[
        1-\omega_2\left(1+\frac{\ell}{q-\ell}\right) > 0,
    \]
    so $g(\rho-\eta,\omega_2) > g(\omega_1,\omega_2)$. By reversing the roles of $\omega_2$ and $\omega_1$ we find $g(\rho-\eta,\rho-\eta) \geq g(\rho-\eta,\omega_2)$ as well, establishing \cref{eq:sec-ineq-helper}. 
\end{proof}

We can now state and prove a lemma bounding the probability that a nontrivial linear combination of two uniform samples from a list-recovery ball lands in a shift of the list-recovery ball. Then, in \cref{cor:mixing-list-rec-ball}, we show how to set parameters to turn this into a statement about the $\delta$-mixing of a list-recovery ball.

\begin{lemma}\label{lem:mixing-list-rec-ball}
	Let $n \in \N$, $q$ a prime power, $1 \leq \ell \leq q$ an integer, and let $\rho \in (0,1-\ell/q)$. Let $T_1,\dots,T_n \subseteq \F_q$ be subsets, each of size $\ell$. Fix $\eta>0$ small enough so that 
    \[
        1-\rho > (1-\rho+\eta)^2 + (\rho-\eta)^2\cdot\frac{\ell}{q-\ell}.
    \]
    Let $\alpha,\beta \in \F_q\setminus\{0\}$ and $y \in \F_q^n$, and let $X,X' \sim B_\rho(T_1 \times \cdots \times T_n)$. Then,
    \begin{align}
        \Pr[\alpha X + \beta X' &\in y+B_\rho(T_1 \times \cdots \times T_n)] \nonumber \\ 
        &\leq 2\sqrt{2n} \cdot \exp_q\left(n(h_{q,\ell}(\rho-\eta)-h_{q,\ell}(\rho))\right) \nonumber \\
        &\qquad + 2n \cdot \exp_q\left(-nD_q\left(1-\rho\|(1-\rho+\eta)^2+(\rho-\eta)^2\frac{\ell}{q-\ell}\right)\right) \label{eq:list-rec-bound-RHS}.
    \end{align}
\end{lemma}

\begin{remark}
    \em
   Since \Cref{lem:omega-ineq} implies
    \[
        1-\rho > (1-\rho)^2 + \rho^2 \cdot \frac{\ell}{q-\ell} \iff 0 < \rho < 1-\frac{\ell}{q},
    \]
    it follows that one can indeed choose $\eta$ small enough to ensure 
    \[
        1-\rho > (1-\rho-\eta)^2 + (\rho+\eta)^2\cdot\frac{\ell}{q-\ell}.
    \]
    Later in \Cref{cor:mixing-list-rec-ball} we will show how to choose $\eta>0$ to bound the two terms in~\cref{eq:list-rec-bound-RHS} by $q^{-\delta n}$ for a concrete $\delta>0$ (which depends on $\rho,\ell,q$).
\end{remark}

\begin{proof}
    Denote $B = B_\rho(T_1\times \cdots \times T_n)$, and fix any nonzero $\alpha_1,\alpha_2 \in \F_q$, and $y \in \F_q^n$. Fix some $\eta \in (0,\rho)$ soon to be determined. Let $Z_1 \sim \alpha_1 \cdot B$, and $Z_2 \sim \alpha_2 \cdot B$. The strategy is to ``massage'' the distribution of $(Z_1,Z_2)$ so that it is easier to analyze the probability that $Z_1 + Z_2 \in B+y$, and then apply \Cref{lem:mixing-errors-helper}. 
	
	For each $i \in [n]$ and $j=1,2$, let $W_{j,i} = 1_{\{Z_{j,i} \notin \alpha_i \cdot T_i\}}$, and let $W_j = (W_{j,1},\dots,W_{j,n})$ for $j=1,2$. By the law of total probability, 
    \begin{align}
        \Pr[Z_1+Z_2 \in y+B] &= \sum_{1\leq w_1,w_2 \leq \lfloor \rho n\rfloor }\Pr[Z_1+Z_2 \in y+B|\wt(W_1)=w_1,\wt(W_2)=w_2]\nonumber \\
        &\qquad\qquad\qquad\qquad\cdot \Pr[\wt(W_1)=w_1,\wt(W_2)=w_2] \nonumber \\ 
        &\leq \sum_{\lceil (\rho-\eta) n\rceil \leq w_1,w_2 \leq \lfloor \rho n\rfloor }\Pr[Z_1+Z_2 \in y+B|\wt(W_1)=w_1,\wt(W_2)=w_2]\nonumber \\
        &\qquad\qquad\qquad\qquad\cdot \Pr[\wt(W_1)=w_1,\wt(W_2)=w_2] \nonumber \\ 
        &\qquad\qquad + \Pr[\wt(W_1) \leq (\rho-\eta)n ~\lor~\wt(W_2)\leq(\rho-\eta)n] \nonumber \\
        &\leq \max_{\lceil (\rho-\eta) n\rceil \leq w_1,w_2 \leq \lfloor \rho n\rfloor} \Pr[Z_1+Z_2 \in y+B|\wt(W_1)=w_1,\wt(W_2)=w_2]\label{eq:start-total-prob} \\ 
        &\qquad\qquad + \Pr[\wt(W_1) \leq (\rho-\eta)n ~\lor~\wt(W_2)\leq(\rho-\eta)n] \ .\nonumber 
    \end{align}
    We first bound
    \begin{align}
        \Pr[\wt(W_1) \leq (\rho-\eta)n ~\lor~&\wt(W_2)\leq(\rho-\eta)n] \leq 2\cdot \frac{|B_{\rho-\eta}(T_1\times \cdots \times T_n)|}{|B_{\rho}(T_1\times \cdots \times T_n)|} \nonumber \\
        &\leq 2\cdot \frac{q^{h_{q,\ell}(\rho-\eta)n}}{\frac{1}{\sqrt{2n}}q^{nh_{q,\ell}(\rho)}} =2\sqrt{2n} \cdot q^{(h_{q,\ell}(\rho-\eta)-h_{q,\ell}(\rho))n} \ , \label{eq:ball}
    \end{align}
    where we used \cref{prop:estimate}. 
    
    We now focus on proving that \cref{eq:start-total-prob} is small. Fix any $w_1,w_2 \in \{\lceil (\rho-\eta) n\rceil,\lceil (\rho-\eta) n\rceil+1,\dots, \lfloor \rho n\rfloor\}$ which maximize $\Pr[Z_1+Z_2 \in y+B|\wt(W_1)=w_1,\wt(W_2)=w_2]$, and set $\omega_j = w_j/n$ for $j=1,2$. We now introduce a ``model'' distribution for $Z_1+Z_2|(\wt(W_1)=w_1,\wt(W_2)=w_2)$ which is easier to analyze. For $i=1,\dots,n$ and $j=1,2$, let $A_{j,i} \in \{0,1\}$ denote i.i.d.\ draws from a $\mathrm{Ber}(\omega_j)$ random variable, and set $A_j = (A_{j,1},\dots,A_{j,n})$. 
    
    Next, define for $i=1,\dots,n$ the following random variables $D_i$: 
    \[ D_i = 
    \begin{cases}
        \alpha_1 \cdot X + \alpha_2 \cdot X' & A_{1,i}=A_{2,i}=0, \\
        \alpha_1 \cdot X + Y' & A_{1,i}=0, A_{2,i}=1, \\
        Y + \alpha_2 \cdot X' & A_{1,i}=1, A_{2,i}=0, \\
        Y + Y' & A_{1,i}=A_{2,i}=1.
    \end{cases}
    \]
    Above, $X,X' \sim T_i$, $Y \sim \F_q \setminus \alpha_1 T_i$, and $Y' \sim \F_q \setminus \alpha_2 T_i$, and they are all independent. Let $D = (D_1,\dots,D_n)$, and note that 
    \begin{align*}
        \Pr[Z_1+Z_2 \in & ~y+B|\wt(Y_1)=w_1,\wt(Y_2)=w_2] = \Pr[D \in y+B|\wt(A_1)=w_1,\wt(A_2)=w_2]\\
        &=\frac{\Pr[D \in y+B,\wt(A_1)=w_1,\wt(A_2)=w_2]}{\Pr[\wt(A_1)=w_1,\wt(A_2)=w_2]}\\
        &\leq \frac{\Pr[D \in y+B]}{\Pr[\wt(A_1)=w_1,\wt(A_2)=w_2]}\\
        &\leq 2n \cdot \Pr[D \in y+B],
    \end{align*}
    where the ultimate inequality follows from \cref{prop:prob-bin-hits-expectation}.

      We now focus on upper-bounding $\Pr[D \in y+B]$. Define the indicator random variables $E_i = 1_{\{D_i \in y_i+T_i\}}$ and note that $D \in y+B \iff \sum_{i=1}^n E_i \geq (1-\rho)n$. Using the above notation, we find
    \begin{align*}
    \Pr[E_i = 1] &= \Pr[D_i \in y+T_i] \\
        &= \Pr[D_i \in y+T_i | A_{1,i}=A_{2,i}=0] \Pr[A_{1,i}=A_{2,i}=0] \\ 
        & \qquad + \Pr[D_i \in y+T_i | A_{1,i}=0, A_{2,i}=1] \Pr[A_{1,i}=0, A_{2,i}=1] \\
        & \qquad + \Pr[D_i \in y+T_i | A_{1,i}=1, A_{2,i}=0] \Pr[A_{1,i}=1, A_{2,i}=0] \\
        & \qquad + \Pr[D_i \in y+T_i | A_{1,i}=A_{2,i}=1] \Pr[A_{1,i}=A_{2,i}=1] \\
        &= \Pr[\alpha_1 \cdot X + \alpha_2 \cdot X' \in y+T_i] \cdot (1-\omega_1)(1-\omega_2) \\ 
        & \qquad + \Pr[\alpha_1 \cdot X+Y' \in y+T_i | A_{1,i}=0, A_{2,i}=1] \cdot (1-\omega_1)\omega_2 \\
        & \qquad + \Pr[Y+\alpha_2 \cdot X' \in y+T_i | A_{1,i}=1, A_{2,i}=0] \cdot \omega_1 (1-\omega_2) \\
        & \qquad + \Pr[Y+Y' \in y+T_i | A_{1,i}=A_{2,i}=1] \cdot \omega_1\omega_2.
    \end{align*}
    
    Applying \Cref{lem:mixing-errors-helper} with $A \leftarrow \alpha_1 T_i$ and $B \leftarrow \alpha_2 T_i$, noticing that $\omega_1, \omega_2 < 1-\ell/q$, we readily get
    \[
        \Pr[E_i = 1] \le (1-\omega_1)(1-\omega_2) + \omega_1\omega_2\cdot\frac{\ell}{q-\ell} \ .
    \]
    Thus, if $F \sim \mathrm{Binomial}(n,(1-\omega_1)(1-\omega_2) + \omega_1\omega_2\cdot\frac{\ell}{q-\ell})$, then the probability that $\sum_{i=1}^nE_i \geq (1-\rho)n$ is at most the probability that $F \geq (1-\rho)n$. By the classical Chernoff--Hoeffding bound (\cref{prop:2-ary-bin-tail}), 
    \[
        \Pr[F \geq (1-\rho)n] \leq \exp_q\left(-nD_q\left(1-\rho\Big\|(1-\omega_1)(1-\omega_2) + \omega_1\omega_2\cdot\frac{\ell}{q-\ell}\right)\right).
    \]
    Recall that $\rho-\eta \leq \omega_1,\omega_2\leq \rho$, so by \Cref{lem:omega-ineq},
    \begin{align*}
        (1-\omega_1)(1-\omega_2) + \omega_1\omega_2\cdot\frac{\ell}{q-\ell} \leq (1-\rho+\eta)^2 + (\rho-\eta)^2 \cdot \frac{\ell}{q-\ell} \ .
    \end{align*}
    Thus, using \Cref{claim:div-properties}, it follows that 
    \[
        D_q\left(1-\rho\Big\|(1-\omega_1)(1-\omega_2) + \omega_1\omega_2\cdot\frac{\ell}{q-\ell}\right) \geq D_q\left(1-\rho\Big\|(1-\rho+\eta)^2+(\rho-\eta)^2\cdot \frac{\ell}{q-\ell}\right).
    \]
    It thus follows that 
    \begin{align*}
        \max_{\lceil (\rho-\eta) n\rceil \leq w_1,w_2 \leq \lfloor \rho n\rfloor}& \Pr[Z_1+Z_2 \in y+B|\wt(W_1)=w_1,\wt(W_2)=w_2] \\
        &\leq 2n \cdot \Pr[D \in y+B]\\
        &\leq 2n\cdot \exp_q\left(D_q\left(1-\rho\Big\|(1-\rho+\eta)^2+(\rho-\eta)^2\cdot \frac{\ell}{q-\ell}\right)\right),
    \end{align*}
    which, when combined with the bound of~\cref{eq:ball}, yields the bound claimed in the lemma.
\end{proof}

\begin{corollary}\label{cor:mixing-list-rec-ball}
    Let $n \in \N$, $q$ a prime power, $1 \leq \ell < q$ an integer, and let $\rho \in (0,1-\ell/q)$.  Let $T_1,\ldots,T_n \subseteq \F_q$, each of size $\ell$. The list-recovery ball $B_{\rho}(T_1 \times \ldots \times T_n)$ is $\delta$-mixing, for
    \begin{equation*}
        \delta=\log_q\left(\frac{(q-\ell)(1-\rho)}{\rho \ell}\right)\cdot \frac{\rho^4 (1-\ell/q-\rho)^2}{16\log q}
    \end{equation*}
    assuming $n \ge \left( \frac{\log q}{\rho(1-\ell/q-\rho)} \right)^{c}$ for some universal constant $c$. 
\end{corollary}

\begin{proof}
    By \Cref{lem:mixing-list-rec-ball}, it suffices to show that we can choose $\eta>0$ so that 
    \[
        1-\rho > (1-\rho+\eta)^2 + (\rho-\eta)^2 \cdot \frac{\ell}{q-\ell} \ 
    \]
    and, that assuming $n \geq \left(\frac{\log q}{\rho(1-\ell/q-\rho)}\right)^c$  for universal constant $c$, we have 
    \begin{align}
        (2\sqrt{2n}) &\cdot \exp_q\left(n(h_{q,\ell}(\rho-\eta) - h_{q,\ell}(\rho))\right) \label{eq:hql-term}\\
        &\qquad+ 2n\exp_q\left(-nD_q\left(1-\rho\|(1-\rho+\eta)^2+(\rho-\eta)^2 \cdot \frac{\ell}{q-\ell}\right)\right) \label{eq:binom-term} \\
        &\leq q^{-n\delta} \nonumber,
    \end{align}
    for $\delta$ as in the statement. 

    Set $\gamma = \frac{\ell}{q-\ell}$.
    Note that then $0<\rho<1-\ell/q$ is equivalent to $0<\rho<\frac{1}{1+\gamma}$.
    Define $\Delta$ so that 
     \[
        1-\rho-\Delta = (1-\rho+\eta)^2 +(\rho-\eta)^2 \cdot \gamma,
    \]
    which means that
    \begin{align}
    \Delta &= (1+2\eta+2\eta\gamma)\rho - (1+\gamma)\rho^2 - 2\eta - \eta^2 - \eta^2 \gamma \nonumber \\
    &\ge (1+2\eta+2\eta\gamma)\rho - (1+\gamma)\rho^2 - (3+\gamma)\eta.\label{eq:Delta-LB}
    \end{align}
    Rearranging, we get that \cref{eq:Delta-LB} is at least $\eta \rho$ whenever
    \begin{equation*}
        0<\eta \leq  \frac{\rho-\rho^2(1+\gamma)}{(3+\gamma)-\rho(1+2\gamma)}.
    \end{equation*}
    Given this, we may choose
    \begin{equation*}
        \eta = \frac{\rho(1-\ell/q-\rho)}{3}>0.
    \end{equation*}
    Indeed, recalling that $\gamma=\frac{\ell}{q-\ell}$, we have
    \begin{equation*}
        \frac{\rho-\rho^2(1+\gamma)}{(3+\gamma)-\rho(1+2\gamma)}\geq \frac{\rho-\rho^2(1+\gamma)}{3+\gamma} = \frac{\rho(q(1-\rho)-\ell)}{3q-2\ell} \geq \frac{\rho(q(1-\rho)-\ell)}{3q} = \eta.
    \end{equation*}
    
    By \Cref{lem:pinsker}, $D(1-\rho\|1-\rho-\Delta) \leq 2\Delta^2$, and so $D_q(1-\rho\|1-\rho-\Delta) \leq \frac2{\ln q} \Delta^2$. 
    We therefore have
    \begin{align}
        2n\exp_q&\left(-nD_q\left(1-\rho\|(1-\rho+\eta)^2+(\rho-\eta)^2 \cdot \frac{\ell}{q-\ell}\right)\right)\le 2n\exp_{q}\left( -\frac{2}{\ln q}\eta^2\rho^2 n \right) \nonumber \\ 
        &\qquad\le \exp_{q}\left(- \frac{ \rho^4 (1-\ell/q-\rho)^2}{16\log q} \cdot n\right),
        \label{eq:case-comb}
    \end{align}
    assuming that $n \ge \poly(\log q,1/\rho, 1/(1-\ell/q - \rho))$. 

    We now obtain a lower bound on $ h_{q,\ell}(\rho) - h_{q,\ell}(\rho-\eta)$ with our chosen $\eta$. \cref{lem:bound-diff-entropy} tells us that
    \begin{equation*}
        h_{q,\ell}(\rho) - h_{q,\ell}(\rho-\eta) \ge \log_q\left(\frac{(q-\ell)(1-\rho)}{\rho \ell}\right)\eta = \log_q\left(\frac{(q-\ell)(1-\rho)}{\rho \ell}\right) \cdot  \frac{\rho(1-\ell/q-\rho)}{3},
    \end{equation*}
    which is positive whenever $\rho < 1-\ell/q$.
    Thus, 
    \begin{multline}\label{eq:case-comb2}
        2\sqrt{2n} \cdot \exp_{q}\left(-\log_q\left(\frac{(q-\ell)(1-\rho)}{\rho \ell}\right)   \frac{\rho(1-\ell/q-\rho)}{3}\cdot  n \right)
    \\
    \le \exp_{q}\left(- \log_q\left(\frac{(q-\ell)(1-\rho)}{\rho \ell}\right)   \frac{\rho(1-\ell/q-\rho)}{6}\cdot  n \right),
    \end{multline}
    where we assumed that $n \ge \poly(\log q,1/\rho, 1/(1-\ell/q - \rho))$. 

    By \cref{eq:case-comb,eq:case-comb2}, it follows that a rough, but sufficient, bound on \eqref{eq:hql-term}+\eqref{eq:binom-term} is 
    \[
        \exp_q\left(-\log_q\left(\frac{(q-\ell)(1-\rho)}{\rho \ell}\right)\cdot \frac{\rho^4 (1-\ell/q-\rho)^2}{16\log q} \cdot n\right),
    \]
    completing the proof.
\end{proof}

\subsection{List Recovery from Errors}

Having established that the list-recovery ball is $\delta$-mixing, we can
repeat roughly the same argument we had for list recovery with erasures, using \cref{lem:main-erasures}.

\begin{theorem}[list recovery with errors]\label{thm:errors}
For every $n \in \mathbb{N}$, a prime power $q$, an integer $\ell < q$, $\rho \in (0,1-\ell/q)$, and $\eps \in (0,1)$, the following holds. With probability at least $1-q^{-n}$, a random linear code $\mathcal{C} \subseteq \F_q^n$ of rate
\[
R = 1-h_{q,\ell}(\rho)-\eps
\]
is $(\rho,\ell,L)$-list-recoverable from errors, with
\[
L = C_{\rho,\ell,q} \cdot \frac{1}{\eps},
\]
provided that $n$ is large enough. More concretely, there exists a universal constant $C$ such that 
\[
C_{\rho,\ell,q} \le q^{\ell \cdot \left( \frac{\log q}{\rho(1-\ell/q-\rho)} \right)^{C}}
\] provided that $n \ge \left( \frac{\log q}{\rho(1-\ell/q-\rho)} \right)^{C}$.
\end{theorem}
\begin{proof}
The technique follows the proof of \cref{thm:erasures}, with the necessary changes. We fix some $T_1,\ldots,T_n \subseteq \F_q$, each of size $\ell$, and denote by $T = B_{\rho}(T_1 \times \cdots \times T_n)$ the list-recovery ball. Recall that we wish to show that $|T \cap \mathcal{C}| \le L$. The rest of the argument exactly matches \cref{thm:erasures}, that is, if $|T \cap \mathcal{C}| > L$, there must exist some $\bm{v} \in \mathcal{F}$ for which $\set{\bm{v}} \subseteq T \cap \mathcal{C}$. Three things change:
\begin{enumerate}
    \item Our bound on $\delta$, which we now take from \cref{cor:mixing-list-rec-ball}.
    \item The number of input lists we union-bound over. While in \cref{thm:erasures} it was $\mathcal{T} = \binom{n}{\alpha n}\binom{q}{\ell}^{(1-\alpha)n}$, here it's $\mathcal{T} = \binom{q}{\ell}^n$.
    \item The gap $\eps$ is now with respect to the ``list-recovery with \emph{errors}'' capacity $1-h_{q,\ell}(\rho)$.
\end{enumerate}
Concretely, the bound on $p_1$ now becomes
\[
p_1 \le \mathcal{T} \cdot b^{\star} q^{-a n} \cdot \sum_{b=\log_{q}(L+1)}^{b^{\star}}|T|^b \cdot q^{-(1-R)nb},
\]
and the bound on $p_2$ is now
\[
p_2 \le \mathcal{T} \cdot \sum_{b=b^{\star}+1}^{L+1}|T|^b q^{-(1-R)nb},
\]
and we are interested in bounding $p = p_1+p_2$. We will re-set $a$ and $c$ (which appears in $b^{\star} = \lceil c\ell/\eps \rceil$)  shortly. Recalling that $|T| \le q^{n h_{q,\ell}(\rho)}$ (see \cref{prop:estimate}), the bound on $p_1$ becomes
\begin{align*}
p_1 &\leq q^{n\ell \log_{q}(eq/\ell)+\log_{q}(b^{\star})-an} \cdot \sum_{b=\log_{q}(L+1)}^{b^{\star}}q^{nbh_{q,\ell}(\rho)-(h_{q,\ell}(\rho)+\eps)nb} \\ 
&\le q^{2n\ell - an - \eps n  h_{q,\ell}(\rho)\log_{q}(L+1)} \le q^{(2\ell-a)n}.
\end{align*}
Taking $a = 2\ell+2$, we get $p_1 \le q^{-2n}$.
Moving on to $p_2$, we have
\[
p_2 \le \mathcal{T} \cdot \sum_{b=b^{\star}+1}^{L+1}q^{nbh_{q,\ell}(\rho)-(h_{q,\ell}(\rho)+\eps)nb} \\
\le q^{2n\ell - h_{q,\ell}(\rho)\eps n b^{\star}}.
\]
Setting $c = \frac{1}{4h_{q,\ell}(\rho)}$ we get $p_{2} \le q^{-2n}$, and overall $p \le q^{-n}$, as desired. The output list size is $q^{\frac{8a}{\delta}} \cdot \frac{c\ell}{\eps}$. Using the fact that $\frac{1}{h_{q,\ell}(\rho)} \le \frac{1}{\log q}$, and that $\log_{q}\left( \frac{(q-\ell)(1-\rho)}{\rho\ell} \right) \le \log_{q}(q/\rho)$, the bound on $C_{q,\rho}$ follows.
\end{proof}

\ifauthors{
\section*{Acknowledgements}

This work was done in part while
the authors were visiting the Simons Institute for the Theory of Computing, supported by DOE grant \#DE-SC0024124. D.\ Doron is supported in part by NSF-BSF grant \#2022644. J.\ Mosheiff is supported by Israel Science Foundation grant 3450/24 and an Alon Fellowship. N.\ Resch is supported by an NWO (Dutch Research Council) grant with number C.2324.0590.
J.\ Ribeiro is supported by
FCT/MECI through national funds and when applicable co-funded EU funds under UID/50008: Instituto de Telecomunicações.
}
\fi

\bibliographystyle{alpha}
\bibliography{refs}

\newcommand{\etalchar}[1]{$^{#1}$}
\begin{thebibliography}{GKO{\etalchar{+}}18}

\bibitem[BDT20]{BDT20}
Avraham {Ben-Aroya}, Dean Doron, and Amnon {Ta-Shma}.
\newblock Near-optimal erasure list-decodable codes.
\newblock In {\em 35th Computational Complexity Conference (CCC)}, pages 1:1--1:27. Schloss Dagstuhl -- Leibniz-Zentrum f{\"u}r Informatik, 2020.

\bibitem[CT06]{CT12}
Thomas~M. Cover and Joy~A. Thomas.
\newblock {\em Elements of Information Theory}.
\newblock John Wiley \& Sons, 2nd edition, 2006.

\bibitem[CZ24]{CZ2024}
Yeyuan Chen and Zihan Zhang.
\newblock Explicit folded {Reed-Solomon} and multiplicity codes achieve relaxed generalized {Singleton} bounds, 2024.
\newblock \url{https://arxiv.org/abs/2408.15925}.

\bibitem[DMOZ22]{DMOZ22}
Dean Doron, Dana Moshkovitz, Justin Oh, and David Zuckerman.
\newblock Nearly optimal pseudorandomness from hardness.
\newblock {\em Journal of the ACM}, 69(6), November 2022.

\bibitem[DW22]{DW22}
Dean Doron and Mary Wootters.
\newblock High-probability list-recovery, and applications to heavy hitters.
\newblock In {\em 49th International Colloquium on Automata, Languages, and Programming (ICALP)}, pages 55:1--55:17. Schloss Dagstuhl -- Leibniz-Zentrum f{\"u}r Informatik, 2022.

\bibitem[Eli57]{Elias1957}
Peter Elias.
\newblock List decoding for noisy channels, 1957.
\newblock \url{https://dspace.mit.edu/handle/1721.1/4484}.

\bibitem[GHK11]{GHK11}
Venkatesan Guruswami, Johan Håstad, and Swastik Kopparty.
\newblock On the list-decodability of random linear codes.
\newblock {\em IEEE Transactions on Information Theory}, 57(2):718--725, 2011.
\newblock Preliminary version at STOC 2010.

\bibitem[GI02]{GI02}
Venkatesan Guruswami and Piotr Indyk.
\newblock Near-optimal linear-time codes for unique decoding and new list-decodable codes over smaller alphabets.
\newblock In {\em 34th Annual Symposium on Theory of Computing (STOC)}, pages 812--821. ACM, 2002.

\bibitem[GI03]{GI03}
Venkatesan Guruswami and Piotr Indyk.
\newblock Linear time encodable and list decodable codes.
\newblock In {\em 35th Annual Symposium on Theory of Computing (STOC)}, pages 126--135. ACM, 2003.

\bibitem[GI04]{GI04}
Venkatesan Guruswami and Piotr Indyk.
\newblock Efficiently decodable codes meeting {G}ilbert-{V}arshamov bound for low rates.
\newblock In {\em 15th Annual Symposium on Discrete Algorithms (SODA)}, pages 756--757. SIAM, 2004.

\bibitem[GI05]{GI05}
Venkatesan Guruswami and Piotr Indyk.
\newblock Linear-time encodable/decodable codes with near-optimal rate.
\newblock {\em IEEE Transactions on Information Theory}, 51(10):3393--3400, 2005.

\bibitem[GKO{\etalchar{+}}18]{GKO+18}
Sivakanth Gopi, Swastik Kopparty, Rafael Oliveira, Noga Ron-Zewi, and Shubhangi Saraf.
\newblock Locally testable and locally correctable codes approaching the {Gilbert-Varshamov} bound.
\newblock {\em IEEE Transactions on Information Theory}, 64(8):5813--5831, 2018.

\bibitem[GLM{\etalchar{+}}22]{GLMRSW22}
Venkatesan Guruswami, Ray Li, Jonathan Mosheiff, Nicolas Resch, Shashwat Silas, and Mary Wootters.
\newblock Bounds for list-decoding and list-recovery of random linear codes.
\newblock {\em IEEE Transactions on Information Theory}, 68(2):923--939, 2022.

\bibitem[GLPS17]{GLPS17}
Anna~C. Gilbert, Yi~Li, Ely Porat, and Martin~J. Strauss.
\newblock For-all sparse recovery in near-optimal time.
\newblock {\em ACM Transactions on Algorithms}, 13(3):1--26, 2017.

\bibitem[GLS{\etalchar{+}}21]{GLS+21}
Zeyu Guo, Ray Li, Chong Shangguan, Itzhak Tamo, and Mary Wootters.
\newblock Improved list-decodability and list-recoverability of {Reed-Solomon} codes via tree packings.
\newblock In {\em 62nd Annual Symposium on Foundations of Computer Science (FOCS)}, pages 708--719. {IEEE}, 2021.

\bibitem[GR18]{GR18}
Venkatesan Guruswami and Nicolas Resch.
\newblock On the list-decodability of random linear rank-metric codes.
\newblock In {\em International Symposium on Information Theory (ISIT)}, pages 1505--1509. IEEE, 2018.

\bibitem[GST22]{GST2022}
Eitan Goldberg, Chong Shangguan, and Itzhak Tamo.
\newblock Singleton-type bounds for list-decoding and list-recovery, and related results.
\newblock In {\em International Symposium on Information Theory (ISIT)}, pages 2565--2570. IEEE, 2022.

\bibitem[Gur03]{Gur03}
Venkatesan Guruswami.
\newblock List decoding from erasures: Bounds and code constructions.
\newblock {\em IEEE Transactions on Information Theory}, 49(11):2826--2833, 2003.

\bibitem[GUV09]{GUV2009}
Venkatesan Guruswami, Christopher Umans, and Salil~P. Vadhan.
\newblock Unbalanced expanders and randomness extractors from {{Parvaresh-Vardy}} codes.
\newblock {\em Journal of the ACM}, 56(4):20:1--20:34, 2009.

\bibitem[HIOS15]{HIOS15}
Iftach Haitner, Yuval Ishai, Eran Omri, and Ronen Shaltiel.
\newblock Parallel hashing via list recoverability.
\newblock In {\em Advances in Cryptology -- CRYPTO 2015}, pages 173--190. Springer Berlin Heidelberg, 2015.

\bibitem[HLR21]{HLR21}
Justin Holmgren, Alex Lombardi, and Ron~D. Rothblum.
\newblock {Fiat–Shamir} via list-recoverable codes (or: parallel repetition of {GMW} is not zero-knowledge).
\newblock In {\em 53rd Annual Symposium on Theory of Computing (STOC)}, pages 750--760. ACM, 2021.

\bibitem[HRW20]{HRW2020}
Brett Hemenway, Noga {Ron-Zewi}, and Mary Wootters.
\newblock Local list recovery of high-rate tensor codes and applications.
\newblock {\em SIAM Journal on Computing}, 49(4):FOCS17--157, January 2020.

\bibitem[INR10]{INR10}
Piotr Indyk, Hung~Q.\ Ngo, and Atri Rudra.
\newblock Efficiently decodable non-adaptive group testing.
\newblock In {\em 21st Annual Symposium on Discrete Algorithms (SODA)}, pages 1126--1142. SIAM, 2010.

\bibitem[KM25]{KM2025}
Sergey Komech and Jonathan Mosheiff.
\newblock Let's have both! {Optimal} list-recoverability via alphabet permutation codes, 2025.
\newblock \url{https://www.arxiv.org/abs/2502.05858}.

\bibitem[KMRS17]{KMRSZ17}
Swastik Kopparty, Or~Meir, Noga {Ron-Zewi}, and Shubhangi Saraf.
\newblock High-rate locally correctable and locally testable codes with sub-polynomial query complexity.
\newblock {\em Journal of the ACM}, 64(2):1--42, 2017.

\bibitem[KRSW18]{KRS+2018}
Swastik Kopparty, Noga {Ron-Zewi}, Shubhangi Saraf, and Mary Wootters.
\newblock Improved decoding of folded {Reed-Solomon} and multiplicity codes.
\newblock In {\em 59th Annual Symposium on Foundations of Computer Science (FOCS)}, pages 212--223. IEEE, 2018.

\bibitem[KT22]{KT22}
Itay Kalev and Amnon {Ta-Shma}.
\newblock Unbalanced expanders from multiplicity codes.
\newblock In {\em Approximation, Randomization, and Combinatorial Optimization. Algorithms and Techniques (APPROX/RANDOM)}, pages 12:1--12:14. Schloss Dagstuhl -- Leibniz-Zentrum f{\"{u}}r Informatik, 2022.

\bibitem[Lev01]{Lev01}
Vsevolod~F. Lev.
\newblock {Linear equations over $\mathbb{F}_p$ and moments of exponential sums}.
\newblock {\em Duke Mathematical Journal}, 107(2):239--263, 2001.

\bibitem[LMS24]{LMS24}
Matan Levi, Jonathan Mosheiff, and Nikhil Shagrithaya.
\newblock Random {Reed-Solomon} codes and random linear codes are locally equivalent, 2024.
\newblock \url{https://arxiv.org/abs/2406.02238}.

\bibitem[LNNT16]{LNN+2016}
Kasper~Green Larsen, Jelani Nelson, Huy~L. Nguyen, and Mikkel Thorup.
\newblock Heavy hitters via cluster-preserving clustering.
\newblock In {\em 57th Annual Symposium on Foundations of Computer Science (FOCS)}, pages 61--70. IEEE, 2016.

\bibitem[LP20]{LP2020}
Ben Lund and Aditya Potukuchi.
\newblock On the list recoverability of randomly punctured codes.
\newblock In {\em Approximation, Randomization, and Combinatorial Optimization. Algorithms and Techniques (APPROX/RANDOM)}, pages 30:1--30:11. Schloss Dagstuhl -- Leibniz-Zentrum f{\"u}r Informatik, 2020.

\bibitem[LS25]{LS2025}
Ray Li and Nikhil Shagrithaya.
\newblock Near-optimal list-recovery of linear code families, 2025.
\newblock \url{https://arxiv.org/abs/2502.13877}.

\bibitem[MRSY24]{MRSY24}
Jonathan Mosheiff, Nicolas Resch, Kuo Shang, and Chen Yuan.
\newblock Randomness-efficient constructions of capacity-achieving list-decodable codes, 2024.
\newblock \url{https://arxiv.org/abs/2402.11533}.

\bibitem[NPR11]{NPR11}
Hung~Q.\ Ngo, Ely Porat, and Atri Rudra.
\newblock Efficiently decodable error-correcting list disjunct matrices and applications.
\newblock In {\em International Colloquium on Automata, Languages, and Programming (ICALP)}, pages 557--568. Springer, 2011.

\bibitem[Res20]{Res20}
Nicolas Resch.
\newblock {\em List-Decodable Codes: (Randomized) Constructions and Applications}.
\newblock PhD thesis, Carnegie Mellon University, 2020.
\newblock \url{http://reports-archive.adm.cs.cmu.edu/anon/2020/abstracts/20-113.html}.

\bibitem[RW14]{RW14}
Atri Rudra and Mary Wootters.
\newblock Every list-decodable code for high noise has abundant near-optimal rate puncturings.
\newblock In {\em 46th Annual Symposium on Theory of Computing (STOC)}, pages 764--773. ACM, 2014.

\bibitem[RW18]{RW18}
Atri Rudra and Mary Wootters.
\newblock Average-radius list-recoverability of random linear codes.
\newblock In {\em 29th Annual Symposium on Discrete Algorithms (SODA)}, pages 644--662. SIAM, 2018.

\bibitem[Tam24]{Tamo2024}
Itzhak Tamo.
\newblock Tighter list-size bounds for list-decoding and recovery of folded {Reed-Solomon} and multiplicity codes.
\newblock {\em IEEE Transactions on Information Theory}, 70(12):8659--8668, 2024.

\bibitem[TZ04]{TZ2004}
Amnon {Ta-Shma} and David Zuckerman.
\newblock Extractor codes.
\newblock {\em IEEE Transactions on Information Theory}, 50(12):3015--3025, 2004.

\bibitem[WR61]{WR61}
John Wozencraft and Barney Reiffen.
\newblock {\em Sequential Decoding}.
\newblock MIT Press, 1961.

\bibitem[ZP81]{ZP81}
Victor~Vasilievich Zyablov and Mark~Semenovich Pinsker.
\newblock List concatenated decoding.
\newblock {\em Problemy Peredachi Informatsii}, 17(4):29--33, 1981.

\end{thebibliography}

\end{document}